\documentclass[aps,pra,showpacs,twoside,twocolumn,10pt]{revtex4-1}
\usepackage[colorlinks=true, citecolor=blue, urlcolor=blue ]{hyperref}
\usepackage{times,epsfig,amssymb,amsfonts,amsmath,bm,subfigure,mathtools,amsthm,braket,soul,enumitem, color,physics,graphics,graphicx}
\usepackage[normalem]{ulem}
\newtheorem{theorem}{Theorem}
\newtheorem{corollary}{Corollary}

\begin{document}
\title{
Dimensional gain in sensing through higher-dimensional quantum spin chain
}
\author{Shivansh Singh$^{1,2}$, Leela Ganesh Chandra Lakkaraju$^2$, Srijon Ghosh$^2$, Aditi Sen(De)$^2$}
\affiliation{$^1$ Department of Physical Sciences, Indian Institute of Science Education and Research, Mohali, Punjab - 140306, India}
\affiliation{$^2$ Harish-Chandra Research Institute, A CI of HBNI, Chhatnag Road, Jhunsi, Prayagraj - 211019, India}


\begin{abstract}

Recent breakthroughs in quantum technology pave the way for extensive utilization of higher-dimensional quantum systems, which outperform their qubit counterparts in terms of capabilities and versatility. We present a framework for accurately predicting weak external magnetic fields using a higher-dimensional many-body quantum probe.  We demonstrate that dimension serves as a valuable resource for quantum sensing when a transverse spin-s Ising chain interacts locally with a magnetic field whose strength has to be determined.  We observe the distinct performance of sensors for spin chains with half-integer and integer spins. Furthermore, we highlight that the time duration appropriate for quantum-enhanced sensing increases with the increase of dimension. Additionally, we observe that, in addition to nearest-neighbor interactions, incorporating interactions between the next nearest-neighbor sites increases sensing precision, particularly for spin chains with integer spins. We also prove the dimensional-dependence of the bound on quantum Fisher information which provides the limit on the precision in estimating parameters. 

\end{abstract}

\maketitle

\newpage

\section{Introduction}


Accurate determination of physical quantities is a  crucial endeavour in numerous branches of quantum technology. Quantum metrology protocols have the potential to achieve higher precision in estimating  physical parameters when compared to their classical counterparts \cite{Braunstein1994, giovenneti2006, giovannetti_nature, Sensing_RMP_2017}. They have a wide range of applications including  optical interferometry \cite{caves1981,rafal2015}, photonics \cite{pirandola2018}, gravity \cite{Alessandro2020}, imaging \cite{albarelli2020} and biology \cite{taylor2016}.  It is well established that the Cram{\'e}r-Rao bound, quantified by quantum Fisher information (QFI) \cite{Braunstein1994,wootters_1981, Helstrom1976, caves1981}, limits the precision  of estimating an unknown parameter $\theta$, as measured by the variance, $\delta \theta$. The multipartite entangled state \cite{HHHH_entanglement} with local measurements  can ensure quantum advantage in sensors  \cite{giovannetti_nature, giovenneti2006} when sensitivity can go beyond the standard quantum limit (the shot noise limit), attainable by separable states although not all entangled states are shown to be beneficial \cite{luca_augusto_prl_2009}.  
Specifically, by exploiting the quantumness of the sensor, it is possible to attain a higher precision with quadratic scaling of QFI,  termed as Heisenberg limit (HL) \cite{giovannetti_nature, giovenneti2006, utkarsh2021, louis2022}, than with the classical ones in which QFI scales linearly with the system size.



There are two primary areas of theoretical research in quantum metrology and sensing. On one hand, it deals with the advancement of the basic concepts of metrology, which includes the achievement of  HL  through the uses of quantum resources such as entanglement \cite{giovannetti_nature}, squeezing \cite{caves1981, schnabel2017}, superposition using quantum switch \cite{zhao2020}, diagonal quantum discord \cite{sone2019}, quantum steering \cite{yadin2021,lee2023} and criticality of many-body quantum systems \cite{irenee2018, rams2018} or by performing sequential unsharp measurements without preshared entanglement \cite{victor2022}. It is interesting to note that the system possessing $k$-body interactions (with $k > 2$) \cite{roy2008, Gietka2022, sergio2007} or allowing the sensors in contact with the target field periodically \cite{utkarsh2021, louis2022} can provide bounds in precision which is beyond HL (known as super-Heisenberg limit). Another key direction is to identify suitable quantum systems that can be utilised to construct quantum sensors (QS) that will outperform standard quantum limit (SQL) and provide more precision in parameter estimation, which is also one of the primary goals of this work.  Moreover, the recent developments have also leveraged many-body quantum scars \cite{shane2023}, stark localization \cite{Xingjian2023}, and the transition from localization to delocalization in lattice systems \cite{sahoo2023} as reliable resources for estimation protocols.

The majority of proposed quantum sensors currently revolve around qubit systems, with only a limited number of studies examining the use of higher-dimensional quantum systems, notably qutrits, as probes for quantum sensing \cite{Shlyakhov2018, shane2021, shane2023}. In this paper, we provide a design of a quantum sensor based on the spin-$s$
 quantum spin chain and demonstrate its advantage over qubit-based sensors. Building such higher-dimensional sensors is motivated by the results, which establish that resources in enlarged Hilbert space can provide higher efficiencies for quantum technologies like quantum key distribution \cite{sanders02, quditQtech}, quantum computing \cite{wang2020}, quantum thermal devices like quantum battery \cite{srijon2022} and quantum refrigerator \cite{tanoy2023} as compared to resources in two-dimensional Hilbert space.  Most importantly, qudit systems can be engineered in the laboratory by using photons \cite{photonexp}, ion trap \cite{ionqudits}, nitrogen-vacancy center \cite{NVcentrequdits}, and also  superconducting systems\cite{supercondqudits}.

 


By formulating analytical expressions for both the standard quantum limit and the Heisenberg limit for $\delta \theta$ in arbitrary dimensions, we propose a qudit-based quantum sensing protocol that can beat SQL for a fixed dimension. In particular, the quantum sensor is initially prepared in the canonical equilibrium state of the transverse spin-$s$ nearest-neighbor Ising chain consisting of $N$ spins with a boundary qudit measured in an optimal basis followed by an evolution for a certain time period. Note that the measurement basis is chosen in such a way that the fidelity of the state with maximum QFI gets optimized with the final state after this preparation process. In the second step, the optimal state interacts with the target field whose strength has to be estimated before evolving with the reversed unitary operator and performing measurement on the same qudit where the initial measurement was performed. Notice that the reverse unitary operator is applied to concentrate the information about the parameter to be estimated in the single qudit while the measurement on an optimized basis provides the probability distribution of the target magnetic field required to assess its performance.
We report that the variance of the parameter not only crosses the SQL, but the distance of the variance with the Heisenberg limit of a fixed dimension decreases with the increase of the dimension of the individual site, thereby exhibiting the {\it dimensional gain}.  Remarkably, we observe that the dimensional benefit is more evident in the case of
 half-integer spins (fermions)  compared to that of integer spins (bosons). However, we demonstrate that the addition of next-nearest neighbor interaction to a quantum sensor with integer spins can help to improve the precision so that the advantage with dimension becomes apparent. Moreover, we establish that time can play important roles in quantum sensing in two ways -- firstly, the range of time in evolution in the preparation step where the variance exceeds SQL, increases with dimensions, and secondly, the exponent of evolution-time required to maximize QFI also increases as the dimension increases.

The paper is organized as follows.  In Sec. \ref{sec:set_the_stage}, we derive the modified expressions of SQL and HL for qudit systems and establish a sufficient condition for detecting genuinely multiqudit entangled  states from the perspective of QFI. Quantum sensing protocol used to demonstrate dimensional advantage is prescribed in Sec. \ref{sec:model}. In Sec. \ref{sec:result}, we establish the dimensional gain in QS. Sec. \ref{sec:conclusion} includes concluding remarks.

\section{Precision limits for qudits}
\label{sec:set_the_stage}

In literature, a qubit-based sensor is constructed which can estimate unknown parameters encoded in it. Bounds on the estimation that guarantees quantum advantage and the scaling of quantum Fisher information have also been derived in a two-dimensional scenario \cite{Sensing_RMP_2017}.
Note that QFI and its scaling can predict whether the quantum resources are useful for estimation or not. However, when the local Hilbert space dimension is arbitrary, the parameter estimation is not known. We derive the exact variation of QFI with the increase of the local dimension in the system and provide compact forms for the standard quantum and Heisenberg limits for the arbitrary $d$-dimensional quantum systems. Moreover, we establish a connection between QFI and multipartite entanglement in arbitrary dimensional systems.

We estimate the strength of the magnetic field encoded into the probe state, $\rho$, as a relative phase denoted by $\theta$ by performing positive operator-valued measurement (POVM) on the probe and by obtaining the distribution of the measured quantity. In the quantum domain, the uncertainty in the estimation of parameter, $\delta \theta$, obeys the quantum Cr\`amer-Rao bound \cite{Sensing_RMP_2017}, i.e., 
\begin{equation}
    \delta \theta \ge \frac{1}{\sqrt{\mu F_Q(\rho, G) }},
\end{equation}
where $F_Q(\rho, G)$ is the quantum Fisher information of the probe state $\rho$ with respect to the generator, $G$ of the parameter, $\theta$ and $\mu$ is the number of independent measurements performed. 
To minimize the error thereby maximizing the precision in the estimation, the probe state must be kept in the state with the highest QFI.
QFI is defined by using  $F_Q(\rho, G) = \text{tr}(\rho L^{2}_{\theta})$, where $L_{\theta}$ is the symmetric logarithmic derivative defined as$
    \frac{\partial \rho}{\partial \theta} = L_{\theta}\rho + \rho L_{\theta}.$
For a unitary evolution, $U = e^{-i G\theta}$ where $G$ is the generator of the parameter $\theta$, the symmetric logarithmic derivative (SLD) takes the form as \cite{GToth2018} $
    L_{\theta} = 2i\sum_{k,l}\frac{p_k - p_l}{p_k + p_l}|k\rangle\langle l| \langle k | G | l\rangle ,
$
which leads to the QFI as 
\begin{equation}
    F_Q(\rho, G) = 2\sum_{k,l}\frac{(p_k - p_l)^2}{p_k + p_l}|\langle k | G | l \rangle|^{2},
\end{equation}
where $p_k$ and $|k\rangle$ are the eigenvalues and the eigenvectors of the probe state $\rho$.
For a pure state, $\rho=|\psi_{\theta}\rangle\langle \psi_{\theta}|$, $L_{\theta} = 2\partial_{\theta} \rho$, and  hence QFI simplifies as $F_Q(|\psi_\theta\rangle\langle \psi_\theta|, G) = 4[\langle \partial_{\theta} \psi_\theta | \partial_{\theta} \psi_\theta \rangle + (\langle \psi_\theta| \partial_{\theta} \psi_\theta\rangle)^{2}]$\cite{paris2009}. 

\subsection*{Standard quantum limit (SQL) and  Heisenberg limit (HL) for qudits }

Let us first discuss the maximum achievable accuracy (standard quantum limit) in parameter estimation protocols with the exclusive use of classical resources or separable probe states \cite{giovenneti2006}. 
The single qudit state that achieves the maximum quantum Fisher information is given by \cite{geza_prl_bound_2018}
\begin{equation}
    |\psi_{d}\rangle = \frac{1}{\sqrt{2}} (|0\rangle + |d-1\rangle),
\end{equation}
where $d = 2s+1$ denotes the dimension of the Hilbert space of a single system of spin-$s$. 
The standard quantum limit can now be found via the interaction of the optimal separable state over $N$-qudits $|\Psi^{sep}\rangle = |\psi_{d}\rangle^{\otimes N}$ with the local operator. It can be experimentally realized as an external magnetic field that interacts locally with the state to encode the parameter $\theta$ which is the strength of the field. The initial state is evolved locally for a certain time $t$ with $H_{\theta} = \theta \sum_{i=1}^{N} S^{z}_{i}$, where $S_i^\alpha$ is the $SU(2)$ representation in $d$-dimension. The matrix representation of the spin operators in arbitrary dimension can be written in the computational basis given by $|m\rangle \in \{|0\rangle,|1\rangle,\ldots,|d-1\rangle\}$ as
\begin{align}
\left\langle m^{\prime}\left|S^x\right| m\right\rangle =&\left(\delta_{m^{\prime}, m+1}+\delta_{m^{\prime}+1, m}\right) \nonumber  \\
&\frac{1}{2} \sqrt{s(s+1)-m^{\prime} m}, \nonumber \\
\left\langle m^{\prime}\left|S^y\right| m\right\rangle =&\left(\delta_{m^{\prime}, m+1}-\delta_{m^{\prime}+1, m}\right) \nonumber \\ &\frac{1}{2 \mathrm{i}} \sqrt{s(s+1)-m^{\prime} m}, \nonumber \\
\left\langle m^{\prime}\left|S^z\right| m\right\rangle  =&~\delta_{m^{\prime}, m} m,  \nonumber \\
\text{ and }\left\langle m^{\prime} \left| S^2 \right| m\right\rangle  =& ~\delta_{m^{\prime}, m} s(s+1).
\end{align}
The encoded time-evolved state takes the form
$ |\Psi_{ini}^{sep} (t)\rangle =  e^{-iH_\theta t} |\Psi^{sep}\rangle =  [\frac{1}{\sqrt{2}} (e^{-i\theta s t}|0\rangle + e^{i\theta s t}|d-1\rangle)]^{\otimes N}. $
The information about the parameter $\theta$ is accumulated by performing POVM on individual qudits with the POVM elements being $M^{+} = P \big[|0\rangle + i |d-1\rangle \big/\sqrt{2}]$ and $M^{-} = \mathbb{I} - M^{+}$, where $P[*]=|*\rangle \langle *|$ denotes the projector, such that the probability of projecting the system to the space of $M^+$ for each qudit becomes $ p^{+} = \frac{1 + \sin (2 \theta s t)}{2}$. 
After repeating the measurement on $N$ independent systems, the uncertainty in measuring $\theta$ can be evaluated by
 \begin{equation}
\delta \theta = \frac{\sqrt{p^{+}(1-p^{+})}}{|\frac{\partial p^+}{\partial \theta}|\sqrt{N \mu}}.
\label{del_theta}
\end{equation}
The factor of $\sqrt{N}$ in the denominator is the consequence of the central limit theorem \cite{giovannetti_nature} when the experiment is performed on $N$ independent systems. A few lines of algebra lead to the minimum uncertainty in the estimation of the parameter, $\delta \theta_{SQL} = \frac{1}{2st\sqrt{N\mu}} $, where the number of measurements $\mu$ can be defined as the ratio between total time consumed $t_{all}$ and the time of each round $t$. It leads to the standard quantum limit for the qudit system as it is indeed clear from this expression that increasing the local dimension of the subsystem reduces the value of $\delta \theta$, in comparison with the qubit system by a factor of $2s$ as $\delta \theta_{SQL} \sqrt{t_{all}}  =  \frac{1}{2s\sqrt{N t} }$.

\subsection*{Bounds for quantum Fisher information and criteria for multipartite entanglement in qudit systems}

It was recently shown that QFI can be used as an efficient witness for multipartite entanglement \cite{geza2012, luca_augusto_prl_2009} of an $N$-party state. 
Specifically, an $N$-qubit quantum state is genuinely multipartite entangled (GME) if the QFI corresponding to the local operator such as $J^z= \sum_{i=1}^N S_i^z$ has a lower bound of $(N-1)^2$.
We derive generalized versions of these bounds in the case of multiqudit systems.

\begin{theorem}

QFI with respect to the local operator $J^\alpha= \sum_{i=1}^N S_i^{\alpha}$ $(\alpha \in \{x,y,z\})$ of a quantum state $\rho$, composed of $N$ qudits with local dimension $d = 2s+1$ ($s$ being the quantum number) is upper bounded by
\begin{equation*}
    F_Q(\rho, J^\alpha) \le 4s^2N^2, \end{equation*} 
\begin{equation*}
   \text{and } \sum_\alpha F_Q(\rho, J^\alpha) \le 4sN(sN + 1).
\end{equation*}
\end{theorem}

\begin{proof}
The QFI corresponding to an operator $J^\alpha$  is related to its variance  as
$F_Q(\rho, J^\alpha) \le 4(\Delta J^\alpha)^2$ \cite{Braunstein1994},  where the equality holds only for the pure states. 
Since
\begin{equation}
    (\Delta J^\alpha)^2 \le \langle {J^\alpha}^2\rangle = \sum_{i\ne j} \langle S_i^{\alpha} S_j^{\alpha} \rangle + \sum_{i=j} \langle {S_i^\alpha}^2\rangle,
    \label{eq: spin_sum}
\end{equation}
and by using inequalities $ \langle {S^\alpha_i}^2\rangle \le s^2$
and $ \langle S^\alpha_i S^\alpha_j \rangle \le s^2, $ we have
\begin{equation}
    (\Delta J^\alpha
    )^2 \le N(N-1)s^2 + Ns^2. 
\end{equation}
Hence, the upper-bound on QFI in $d$-dimension becomes
\begin{equation}
    F_Q(\rho, J^\alpha) \le 4s^2N^2.
\end{equation}
The state that saturates the bound is a genuinely multipartite entangled state, given by $|\Psi^{ent}\rangle = \frac{1}{\sqrt{2}} (|0\rangle^{\otimes N} + |d-1\rangle^{\otimes N})$. 
Summing over all the directions of both sides of Eq. (\ref{eq: spin_sum}), i.e., $\sum_{\alpha} F_Q(\rho, J^\alpha) \le 4\sum_{\alpha} (\Delta J^\alpha)^2$, we obtain 
\begin{align}
    \sum_{\alpha} (\Delta J^\alpha)^2 \le &\sum_{i=j} \langle{{S^x}^2}\rangle + \langle{{S^y}^2}\rangle + \langle{{S^z}^2}\rangle + \nonumber \\
    &  \sum_{i\ne j}\langle{S^x_{i} S^x_{j}}\rangle + \langle{S^y_{i} S^y_{j}}\rangle + \langle{S^z_{i} S^z_{j}}\rangle.
\end{align}
Since $\langle{{S^x}^2}\rangle + \langle{{S^y}^2}\rangle + \langle{{S^z}^2}\rangle =  \langle S^2 \rangle = s(s+1)$ and the inequality of the form
$\langle{S^x_iS^x_j}\rangle + \langle{S^y_iS^y_j}\rangle + \langle{S^z_iS^z_j}\rangle \le s^2$ $\forall~ \{i,j\}$ exists, (the latter can be inferred by considering one of the eigenvectors of either of spin operators and the corresponding expectation values of other spin operators being vanishing), 
the bound on the QFI becomes
\begin{equation}
    \sum_\alpha F_Q(\rho, J^\alpha) \le 4Ns(Ns+1).
\end{equation}
\end{proof}
Having obtained the bounds on the QFI, we can use them to obtain bounds for $k-$producible states i.e., where the number of entangled parties is upper bounded by $k$. Consider a $k-$producible pure state 
$|\psi\rangle = \bigotimes_{i = 1}^l |\psi^{r_i} \rangle$ where $|\psi^{r_i} \rangle$ represents the $r_i$-party entangled state and $\max\{ r_1, r_2, \ldots r_l\} = k$ and $\sum_{i = 1}^l r_i = N$.
Now using the additivity property of QFI for separable states and  and the inequalities derived from the variance of $J^\alpha$, we obtain
\begin{equation}\label{qfi_ksep_bound}
    F_Q(\rho, J^\alpha) \le 4s^2\sum_i r_i^2, 
\end{equation}
\begin{equation}
    \text{and }\sum_\alpha F_Q(\rho, J^\alpha) \le 4\sum_i sr_i(sr_i + 1).   
\end{equation} 

\textbf{Heisenberg limit:} Instead of a seperable state, if we now use entangled state to encode $\theta$ using Ramsey interferometry \cite{huelga, giovenneti2006}, we can overcome SQL. Following the Cramer-Rao bound, the uncertainty for the highest QFI state, under the best possible measurement given by the SLD, is $\delta \theta =  \frac{1}{\sqrt{\mu F_Q(|\Psi^{ent}\rangle, ~J^z)}} = \frac{1}{2sNt \sqrt{\mu} }$ as the value of $\theta$ is induced on to the state using the operator $e^{-i\theta J^z t}$. As we did in the case of SQL, the number of independent measurements can be inferred as $\mu = \frac{t_{all}}{t}$, such that the Heisenberg limit for qudit system becomes $\delta \theta_{HL} \sqrt{t_{all}} = \frac{1}{2sN\sqrt{t}}$ \cite{matsuzaki2011, Matsuzaki_2021}.  


\begin{corollary}
    There exist $k$ $-$ producible entangled states with individual subsystems having spin-$s_2$ that have QFI greater than the maximum possible QFI achieved by $N$ - party spin-$s_1$ state, where $s_2 > s_1$, such that $k \ge \left(1+\frac{N s_1}{s_2}\right).$
\end{corollary}

\begin{proof}
    As seen from the proof of Theorem $1$, the $N$-party state that maximizes QFI for spin-$s_1$ is $\frac{1}{\sqrt{2}}(|0\rangle^{\otimes N } + |d_1-1\rangle^{\otimes N })$, where $d_1 = 2s_1 + 1$, for which the QFI is $4s_1^2N^2$. Now, as the spin increases to $s_2$, the $k$-producible entangled states in the corresponding Hilbert space can have QFI which is lower bounded by $4s_2^2(k-1)^2$, such that the maximum QFI of $s_1$-system can be realized by a fewer party entangled state using $s_2$-system. Mathematically, 

    \begin{equation*}
        4s_1^2N^2 \le 4s_2^2(k-1)^2  
       \implies k \ge \left(1+\frac{N s_1}{s_2}\right),
    \end{equation*}
which implies that a $k$-producible state in spin-$s_2$ is enough to generate the same QFI as the maximum QFI of the spin-$s_1$.
    
\end{proof}

Let us illustrate \cite{luca_augusto_prl_2009} that not all $N$-qudit entangled states are useful for quantum metrology. To manifest this, we generate Haar uniformly \cite{Bengtsson_Zyczkowski_2006} pure four-qudit states $\sum_{ijkl = 0}^{d-1}a_{ijkl}|ijkl\rangle$ with $a_{ijkl} = a^{\prime}_{ijkl} + i a^{\prime \prime}_{ijkl}$ where $ a^{\prime}_{ijkl} (a^{\prime \prime}_{ijkl})$ are chosen randomly from Gaussian distribution with zero mean and unit standard deviation, 
calculate QFI as $F_Q(\rho, J^z)$,  and study the corresponding frequency ($\nu$) distribution. In the case of spin-$1/2$, the mean of the distribution is very close to the QFI of the optimal separable state, $F_Q(|\Psi^{sep}\rangle, J^z)$ while as the dimension increases, the mean of the distribution is significantly smaller than $F_Q(|\Psi^{sep}\rangle, J^z)$ (as shown in Fig. \ref{fig:haar_qfi} by dotted lines). This is possible due to the fact that the number of states other than $|0\rangle$ and $|d-1\rangle$ in which the superposition is made up increases, as $s$ increases, which results in a decrement of QFI. This shows that the states that are useful for quantum-enhanced metrology become increasingly sparse as the dimension increases. 

\begin{figure}
    \centering
    \includegraphics[width=\linewidth]{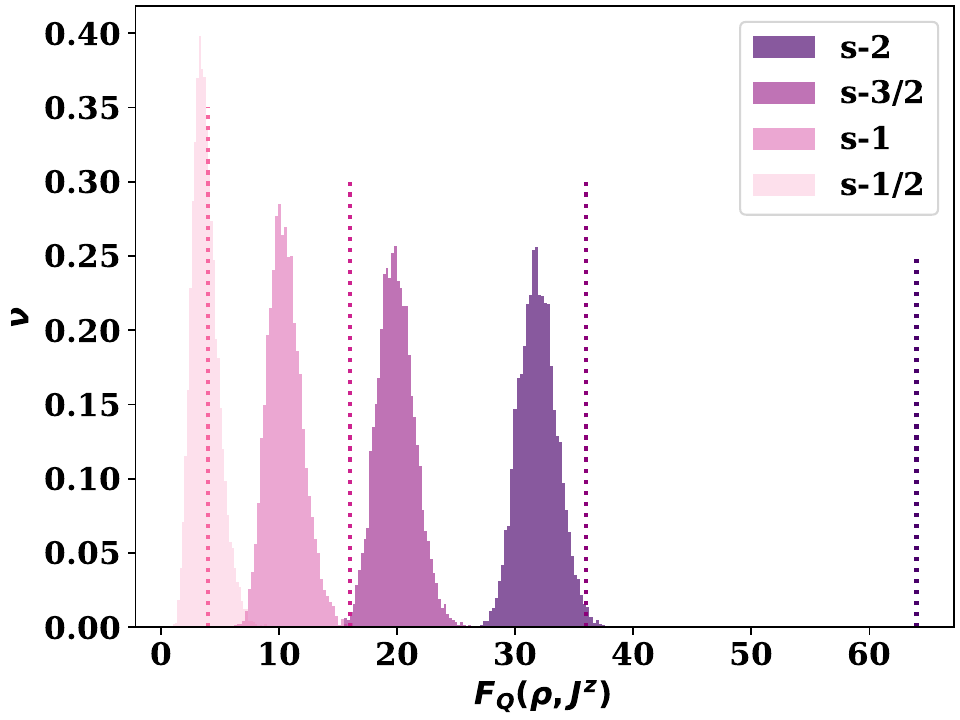}
    \caption{(Color Online.) The normalized frequency $(\nu)$ distribution (ordinate) of  QFI, $F_Q(\rho, J^z)$ (abscissa) of Haar uniformly generated four-qudit states with different values of local dimension, \(d= 2s +1\). The number of random states generated is $10^4$. The dotted lines represent the QFI of the optimal separable state $F_Q(|\Psi^{sep}\rangle, J^z)$. The distribution and SQL become darker as the value of $s$ increases. Both the axes are dimensionless.}
    \label{fig:haar_qfi}
\end{figure}

This motivates us to design a measurement-based sensing protocol that does not require maximum genuine multipartite entanglement. In addition to that, we show that as the dimension of the system increases, the amount of time it takes to reach the maximum QFI or minimum uncertainty also decreases. This two-fold benefit is displayed in the succeeding section which can be termed as dimensional gain.

\section{Minimum uncertainty transverse field spin-s Ising chain}
\label{sec:model}

First, we present the sensing protocol by using a nearest-neighbor transverse field Ising chain, consisting of $N$ number of qudits with open boundary conditions, given by
\begin{equation}
H_{sen} = H_{zz} + H_{f}\\
=J\sum_{i=1}^{N-1} S^z_{i} S^z_{i+1} + h \sum_{i=1}^N S^x_{i},
\end{equation}
where $J$ denotes the coupling constants between the sites and $h$ is the strength of the local magnetic field acted on the individual site. 
The initial state is prepared in the canonical equilibrium state of the Hamiltonian, i.e., $\rho_{ini} = \exp(- \beta H_{sen}) / Z$ with $\beta = \frac{1}{k_B T}$ being the inverse temperature $T$, and $Z = \tr[\exp(- \beta H_{sen})]$ is the partition function.  
By using this as the probe, our main aim is to accurately estimate the strength of the external magnetic field $\omega$, described by the target Hamiltonian
\begin{equation}
    H_{tar} = \omega \sum_{i=1}^N S^z_{i}.
\end{equation}
\noindent
The sensing protocol considered here consists of three basic steps (c.f. \cite{Matsuzaki_2021}):

\textbf{\textit{Step 1} - State preparation and generation of resource.} In the preparation step, we perform a POVM in the optimized basis on a boundary qudit, say the first site of the spin chain, so that, after evolving the post-measured state for a certain time, say, $t^{*}$, the resultant state shares a significantly high fidelity with the maximally QFI state on that specific dimension. Unlike the qubit case, finding the optimal measurement basis is non-trivial in the qudit case. However, if we restrict ourselves to the two-outcome POVM in arbitrary dimensions, we can assume that the measurement is performed in the basis, $M_1 = P[\sum_{i = 0}^{d-1} a_i |i\rangle]$ and $M_2 = \mathbb{I}-M_1$. 
We find the optimal measurement by ensuring that the state at the end of this step as $\rho(t^{*}) = \exp(-i H_{sen} t^{*}) \rho_{M_{1}} \exp(i H_{sen} t^{*})$ with $\rho_{M_1} = \frac{M_1 \rho_{ini} M_1}{\text{tr}(M_1 \rho_{ini} M_1)}$ becomes as close as possible to the maximum QFI state. Due to the nature of the dynamics of the protocol, the last qudit or the qudit on the opposite end of the measured qudit does not entangle with the rest of the chain. Thus, we calculate the fidelity between the maximal QFI state and $\tilde{\rho}(t^*)=\tr_{N}\rho(t^*)$ as $\mathcal{F} = \langle \Psi^{ent}|\tilde{\rho}(t^*)|\Psi^{ent}\rangle$. In this regard, we numerically maximize $\mathcal{F}$ with dimension, $3 \le d \le 6$.  For example, let us consider the spin-$3/2$ system in which an arbitrary projector can be applied on the first site written as $P[a_0 |0\rangle +  a_1 |1\rangle +  a_2 |2\rangle + a_3 |3\rangle]$ where $a_0^2+ a_1^2+ a_2^2+ a_3^2 = 1$. Extensive numerical search confirms that it is sufficient to keep the coefficients to be real and the optimization lies on the surface of $a_0 = a_3$. We then maximize the fidelity $\mathcal{F}$ over the two-dimensional space of $a_1$ and $a_2$ and find that the fidelity $\mathcal{F}$ is maximum only when they vanish (see Fig. \ref{fig:opt_measure}). This indicates that the optimal measurement in this step can be taken as $M_1 = P[(|0 \rangle + |d-1\rangle)/\sqrt{2}] $ and its orthogonal subspace. 

\begin{figure}
    \centering
    \includegraphics[width=\linewidth]{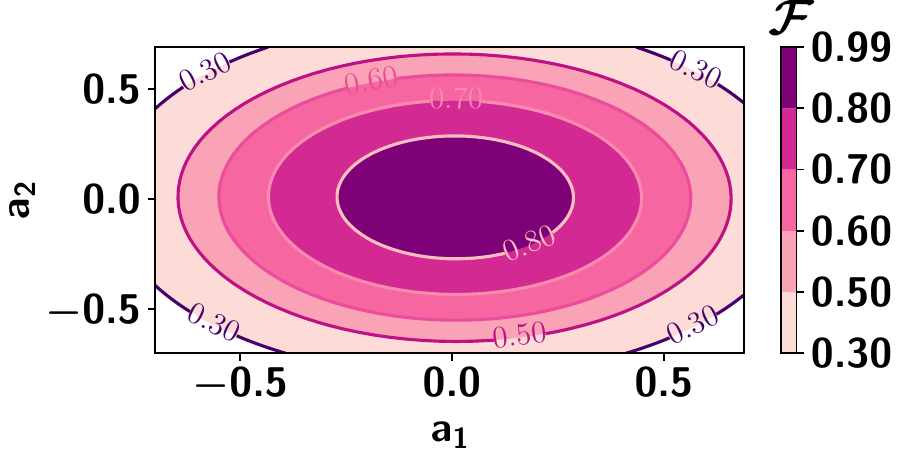}
    \caption{ (Color Online.) The map plot of fidelity $(\mathcal{F})$ of the four-qudit state after Step 1 (having spin-\(3/2\)) with the maximal QFI state when the measurement outcome is $M_1 = P[a_0 |0\rangle +  a_1 |1\rangle +  a_2 |2\rangle + a_3 |3\rangle]$ with $a_0 = a_3$.  The plot is against $a_1$ (horizontal axis) and $a_2$ (vertical axis). The relevant range of both the parameters is in between $[-0.7, 0.7]$. The evolution-time $t^*_{3/2} = 50.76$. It shows that the fidelity becomes maximum when \(a_1= a_2 = 0\). Both the axes are dimensionless. }
    \label{fig:opt_measure}
\end{figure}

 

\textbf{\textit{Step 2} - Encoding the target parameter.} With $\rho(t^*)$, the sensor Hamiltonian  $H_{sen}$ is turned off,  
 allowing the system to interact with the target field for a time interval of $t_{int}$. As a result, the system evolves according to the unitary, $U_{tar}(t_{int}) = e^{-i (H_{tar}) t_{int}}$ which leads to a resultant state $\rho(t^{*}+t_{int}) = U_{tar}(t_{int})\rho(t^*)U_{tar}^\dagger(t_{int})$ which contains the information about $\omega$. 

\textbf{\textit{Step 3} - Estimating $\omega$ through optimized measurement.} To estimate the relevant parameter, $\omega$, the system evolves according to the time-reversed unitary generated by the sensor Hamiltonian $H_{sen}$ up to a time interval $t^*$ as $U_{sen}^\dagger(t^*) = \exp(i H_{sen} t^{*})$. The boundary qudit, on which the measurement is initially performed, becomes disentangled from the rest of the part of the spin-chain, and all the information gets concentrated on that subsystem. We then measure the same qudit in the aforementioned basis $M^+$ and its orthogonal space in obtaining SQL. Incidentally, the measurement is the eigenbasis of the SLD such that the Cramer-Rao bound is saturated. After the post-selection of the measurement output corresponding to $M^+$, we obtain the probability distribution that contains information about the target field. Denoting the total time required to complete the entire protocol as $t_{all} = 2t^* + t_{tar}+$ time for measurements in steps 1 and 3 and initialization in step 1 of the state in steps 1 and 3 and $t_{sense} = 2t^* + t_{tar}$, the uncertainty in estimating the target field $\omega$ reads as

\begin{equation}
\delta \omega \sqrt{t_{all}} = \frac{\sqrt{p^{+}(1-p^{+})}}{|\frac{\partial p^+}{\partial \theta}|} \sqrt{t_{sense}}.
\label{del_omega_tall}
\end{equation}
The value of $\frac{t_{all}}{t_{sense}}$ denotes the independent repetitions of the protocol. For a given dimension or alternatively, spin-$s$ systems, the minimum uncertainty, denoted as $\delta \omega_{\min}$ can be analyzed to determine the performance of the sensing scheme.

\section{Local dimension as Resource for quantum sensing: Integer vs half-integer spin}
\label{sec:result}

We analyze the performance of the sensor constructed via a spin-$s$ transverse  Ising chain. To gauge the accuracy of the proposed sensor, we evaluate the value of $\delta \omega$, i.e., the variance of the target magnetic field, as given in Eq. (\ref{del_omega_tall}). We establish first the dimensional advantage, i.e., the benefits for building spin-$s$ QS $(s > \frac{1}{2})$.
In addition, we also illustrate the advantage in time for the higher dimensional sensor. Interestingly, we report that the dimensional benefits are more pronounced for half-integer QS compared to the integer-spin one. 

\textbf{Dimensional gain.} To establish that the dimension in the individual site is beneficial, one has to compare $\delta \omega \sqrt{t_{all}}$ obtained from the scheme described in Sec. \ref{sec:model} with both $\delta \omega_{SQL}\sqrt{t_{all}}$ and $\delta \omega_{HL} \sqrt{t_{all}}$. In particular, when $\delta \omega \sqrt{t_{all}} < \delta \omega_{SQL} \sqrt{t_{all}}$ (marked with dashed line in Fig. \ref{fig:half_int_time}), it guarantees the quantum advantage. We observe that for a given spin quantum number $s$, there exists a range of evolution-time in Steps $1$ and $3$, $t^*_s$, for which   $\delta \omega \sqrt{t_{all}}$ always goes below $\delta \omega_{SQL}\sqrt{t_{all}}$ and $\delta \omega \sqrt{t_{all}}$ itself decreases with varying dimension (as shown in Fig. \ref{fig:half_int_time}). To establish dimensional gain, we have to compare $\delta \omega \sqrt{t_{all}}$, especially $\delta \omega_{\min} \sqrt{t_{all}} = \min_{t^*_s}  \delta \omega \sqrt{t_{all}} $ with $\delta \omega_{HL}\sqrt{t_{all}}$, i.e., we define a quantity called dimensional gain as $\Delta^{adv} = (\delta \omega_{\min} - \delta \omega_{HL})\sqrt{t_{all}}$. If one can demonstrate that $\Delta^{adv}$ decreases with $s(d)$, which is indeed the case for our scheme as depicted in Fig. \ref{Fig:half_full_fig}, we confirm the dimensional gain.  Note that in Step 1 of Sec. \ref{sec:model}, we use Eq. (\ref{qfi_ksep_bound}) to ensure that the system is at most $(N-1)$-party entangled, due to which we set the Heisenberg limit with respect to the $N-1$ number of qudits as opposed to $N$. Moreover,  we find that the range of $t_s^*$ where $ \delta \omega \sqrt{t_{all}} <  \delta \omega_{SQL} \sqrt{t_{all}}$ increases with the increase of dimension. Specifically, for a fixed $s$, we can define $\Delta t_s^* = (t_s^*)_{\max} - (t_s^*)_{\min} $, where $(t_s^*)_{\max(\min)}$ denotes the maximum (minimum) value of $t_s^*$ up to which the precision is more than that of the SQL. We observe that  $\Delta t_{s_2}^* > \Delta t_{s_1}^*$ where $s_2 > s_1$, thereby again exhibiting the advantage of dimension. Notice that $\Delta t_{s}^*$ can be presented as robustness for the resource to obtain the quantum advantage in sensing which is different than the one reported before \cite{victor2022}. 

Let us now compare the benefits for the systems with integer and half-integer spins. Interestingly, we notice that the case of $s = \frac{n}{2}$ $\forall n = 1,3,\ldots$ $\Delta^{adv}$ is more substantial compared to the case with integer spins (see Fig. \ref{Fig:half_full_fig}). Although we cannot find the clear reasoning behind this, especially by examining the fidelity after Step 1, corresponding to $M^+$, we notice that in the case of half-integer spins, the fidelity is much higher than that of the model with integer spins.  We will exhibit that such differences can be eliminated by adding more interaction between neighboring sites. 

\setlength{\tabcolsep}{12pt} 
\begin{table}
  \centering
  \begin{tabular}{cccc}
    \hline
    $s$ & $\delta\omega_{SQL}\sqrt{t_{all}}$ & $\delta\omega_{HL}\sqrt{t_{all}}$ & $\delta\omega_{min}\sqrt{t_{all}}$ \\
    \hline
    1/2 & 0.0126157 & 0.00841044 & 0.00917038 \\
    1 & 0.00630783 & 0.00420522 & 0.00515845 \\
    3/2 & 0.00420522 & 0.00280348 & 0.00306021 \\
    2 & 0.00315392 & 0.00210261 & 0.00302282 \\
    5/2 & 0.00252313 & 0.00168209 & 0.00188483 \\
    3 & 0.00210261 & 0.00140174 & 0.00196205 \\
    \hline
  \end{tabular}
  \caption{ For each dimension or spin quantum number, $s$, the corresponding SQL, HL, and minimum uncertainty for the sensor are tabulated. All of the uncertainty values are scaled by $\sqrt{t_{all}}$. We consider $ N =4, \beta = 10$, $\omega = 10^{-6}$, $t_{sense} \approx~ t_{int} = 500\pi$, $h/J = 0.10$.}
  \label{tab:min_uncert_table}
\end{table} 


\subsection*{Advantage in time for spin-s quantum sensor }

The total time required to complete the protocol has recently been argued as a resource for quantum sensing  \cite{victor2022}. We propose yet another kind of advantage in time for the spin-s sensing protocol considered here. As shown in the preceding section, the QFI of the state after Step 1 increases as the local dimension increases. We observe that the maximum QFI can be obtained with increasing exponent of time $t^*$ as seen in Fig. \ref{fig:time_scaling}. In particular, we find $\gamma$ in $(t^*)^\gamma$ when QFI approaches maximum increases with the increase of dimension. Even here, there is a distinction between half-integer and integer spins. Specifically, we note that the value of the maximal QFI for integer spins is less than the half-integer ones with respect to the value of minimum uncertainty (see Table. \ref{tab:min_uncert_table}). Although the values of $\gamma$ are comparable for all values of $s$ (as shown in Table \ref{tab:qfi_time}).

\begin{figure}
    \centering
    \includegraphics[width = \linewidth]{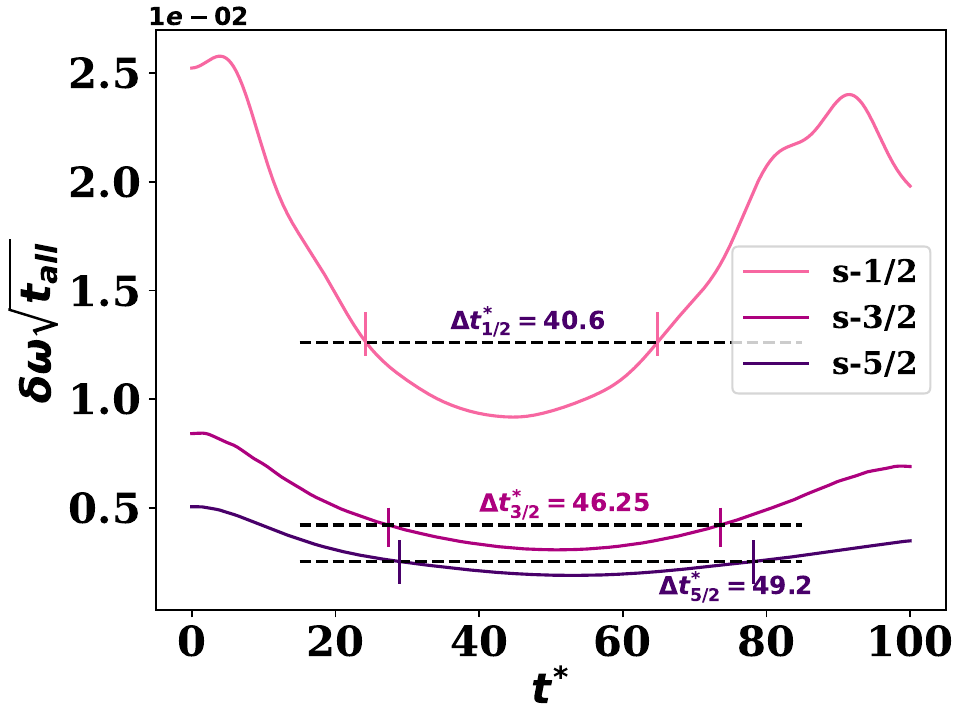}
    \caption{(Color Online.) Variation of minimum uncertainty $\delta \omega \sqrt{t_{all}}$ (ordinate) with the evolution-time $t^*$(abcissa). Different curves correspond to different quantum sensors as described in Sec. \ref{sec:model}, with different values of half-integer spins. The horizontal lines represent the SQL for the corresponding spin quantum number, $s$, given by $\frac{1}{2s\sqrt{N t_{int}}}$. 
    Here $ N = 4, \beta = 10$, $\omega = 10^{-6}$, $t_{sense} \approx t_{int} = 500\pi$ and $h/J = 0.1$. Both the axes are dimensionless.   }
    \label{fig:half_int_time}
\end{figure}
\begin{figure}
    \centering
    \includegraphics[width=\linewidth]{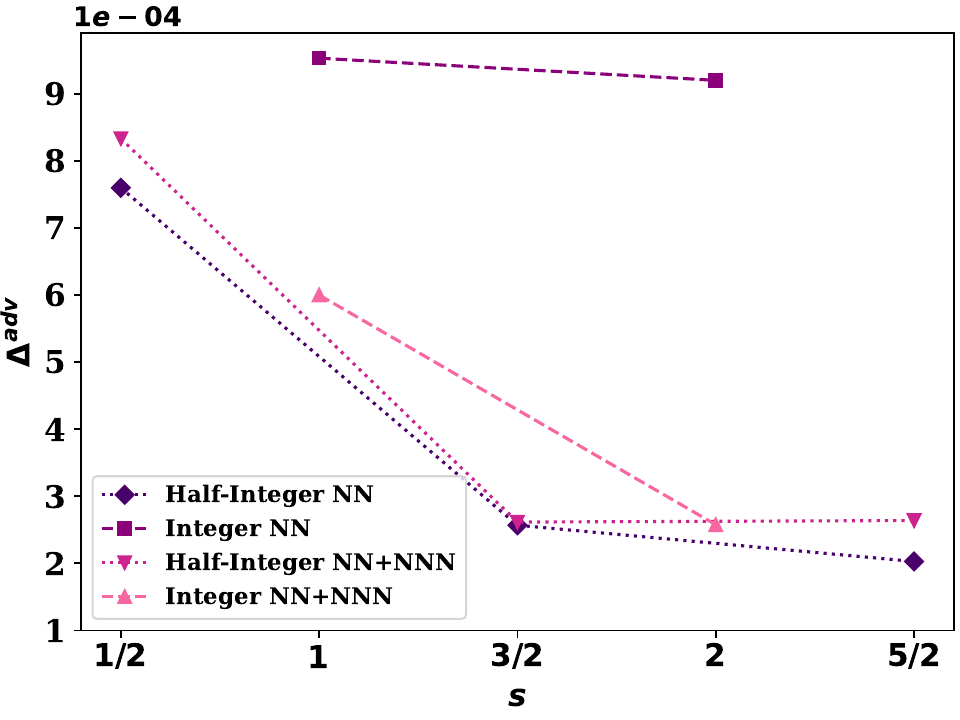}
    \caption{ (Color Online.) Dimensional gain $( \Delta^{adv} = \delta \omega_{min} - \delta\omega_{HL})\sqrt{t_{all}}$ (Ordinate) vs $s$ (abscissa). Sensors comprised of only nearest-neighbor (NN) interaction and nearest and next-nearest-neighbor (NN+NNN) interactions, as in Eq. (\ref{eq:ham_longrange}) are indicated with four- and three-sided polygons respectively. Dashed and dotted lines represent systems with integer and half-integer spins. All other specifications are the same as in Fig. \ref{fig:half_int_time}. In addition, the value of fall-off rate $\alpha$ used in case of \(H_{sen}^{\alpha}\) is optimized for each spin quantum number $s$ to obtain the highest $\Delta^{adv}$ and the optimal pair is denoted by \(\{s, \alpha\}\). In particular, we find $\{\frac{1}{2}, 2.89\}, \{1, 1.76\}, \{\frac{3}{2}, 3.56\}, \{2, 2.39\}$, and $\{\frac{5}{2}, 3.41\} $. Both the axes are dimensionless.}
    \label{Fig:half_full_fig}
\end{figure}

\setlength{\tabcolsep}{15pt} 
\begin{table}[h]
  \centering
  \begin{tabular}{cccc}
    \hline
    $s$ & $\gamma$ \\
    \hline
    1/2 &  1.15743928\\
    1 &  1.46397299\\
    3/2 & 1.34294337 \\
    2 & 1.54826223 \\
    5/2 & 1.65154762 \\
    \hline
  \end{tabular}
  \caption{The quantum fisher information initially increases with $t^*$ and attains a maximum at which $\delta \omega_{min}$ is obtained. The corresponding QFI is fitted with the curve $(t^{*})^{\gamma}$. We observe that $\gamma$ increases with $s$. We consider the same parameters as in Table. \ref{tab:min_uncert_table}.}
  \label{tab:qfi_time}
\end{table}


\begin{figure}
    \centering
    \includegraphics[width = \linewidth]{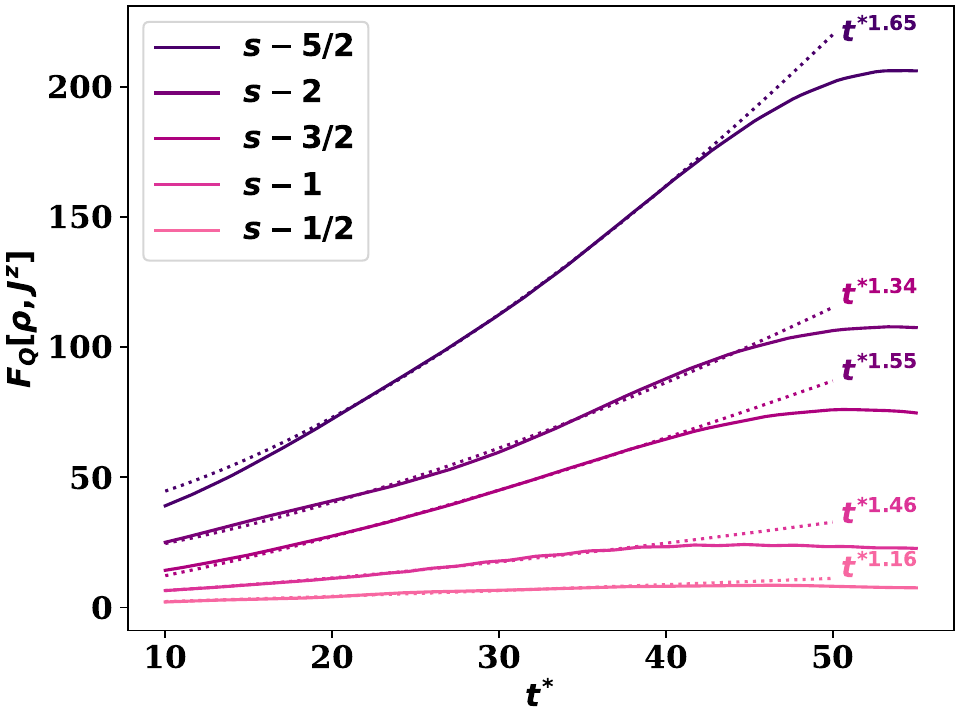}
    \caption{(Color Online.) QFI (vertical axis) with respect to $t^{*}$ (horizontal axis). The dotted lines represent the fit in the evolution-time of the corresponding maximum QFI, i.e. the fit $(t^{*})^{\gamma}$ where $\gamma$ is given in Table. \ref{tab:qfi_time}. The color shade becomes lighter to darker with the increase of dimension. All the specifications are the same as in Fig. \ref{fig:half_int_time}. Both the axes are dimensionless.}
    \label{fig:time_scaling}
\end{figure}




\subsection*{Effects of next-nearest-neighbor interactions on quantum sensors}

Long-range interacting spin systems naturally appear in trapped-ion systems, possessing highly multipartite entangled states and hence have the potential to increase quantum benefits in metrology \cite{monika2023, abolfazl_longrange_2023}.
Apart from this motivation, we also want to raise dimensional gain for the spin model consisting of integer spins. Let us consider the sensor Hamiltonian which is modified by adding the next nearest neighbor term as \cite{debasis_sadhukhan_2020, mimicking_lakkaraju_2022, dqpt_lakkaraju_2023}, 

\begin{equation}
    H_{sen}^\alpha = J\sum_{i=1}^{N-1}S^z_i S^z_{i+1} + \frac{J}{2^\alpha}\sum_{i=1}^{N-2}S^z_{i}S^x_{i+1}S^z_{i+2} + h\sum_{i=1}^{N} S^x_i, 
    \label{eq:ham_longrange}
\end{equation}
where $J =(1+\frac{1}{2^\alpha})^{-1}$ is known as the Kac factor and the fall-off rate $\alpha > 0$ represents the range of interactions. For $s=\frac{1}{2}$, the model can be easily mapped to the spinless fermions using the Jordan-Wigner transformation which leads to the next-nearest neighbor interacting Kitaev chain \cite{vodola2015longrangephases, Vodola1, VodolaThesis}. 
Following the similar steps as above, the sensors with integer spins can provide a significant advantage in terms of $\Delta^{adv}$ with the increase of spin quantum number $s$ while the sensing precision with half-integer spins remains unaltered with the additional interactions (see Fig.\ref{Fig:half_full_fig}).

\section{conclusion}
\label{sec:conclusion}

Quantum sensors (QS) are built to assess unknown parameters that are encoded in them and can provide advantages beyond classical thresholds. We delved into the realm of higher-dimensional quantum sensing, exploring the possibility of acquiring quantum advantage beyond qubit-based ones. To construct a useful quantum sensor with multiqudit systems, the first requirement is to fix the standard in precision, which is provided by the mathematical expressions for the standard quantum limit (SQL) and the Heisenberg limit for arbitrary spin quantum numbers. Additionally, we established an upper bound of quantum Fisher information (QFI) which furnishes the bound on the error in the estimation of the parameter and the connection with multipartite entanglement of the state. By considering a spin-$s$ Ising chain with $s>\frac{1}{2}$ as the quantum sensor, we demonstrated that along with system size, the Hilbert space dimension of local subsystems of QS can also be used as an efficient resource in the sensing protocols. Within this context, we explored the distinctive behaviors of half-integer and integer spins in quantum sensing, providing a promising avenue for investigating the role of indistinguishable particles, especially, the contrasting behaviors of bosons and fermions in parameter estimation scenarios. 

Furthermore, our investigations highlighted the utilities of evolution-time as a resource. On one hand, we found that the range of time-period in which the uncertainty is below SQL increases with the increase of the dimension of individual sites. On the other hand, we observed the scaling of QFI with the evolution-time, when QFI approaches maximum, increases with the augmentation of the local dimension of the subsystems, indicating the eventual attainment of the Heisenberg limit. Our work unfolds new possibilities for leveraging quantum advantages in sensing and for expanding the frontiers of precision measurement in quantum technologies with the variation of dimensions.

\begin{acknowledgements}

We acknowledge the support from the Interdisciplinary Cyber Physical Systems (ICPS) program of the Department of Science and Technology (DST), India, Grant No.: DST/ICPS/QuST/Theme- 1/2019/23. We acknowledge the use of \href{https://github.com/titaschanda/QIClib}{QIClib} -- a modern C++ library for general purpose quantum information processing and quantum computing (\url{https://titaschanda.github.io/QIClib}) and cluster computing facility at Harish-Chandra Research Institute. This research was supported in part by the ’INFOSYS scholarship for senior students’.

\end{acknowledgements}

\bibliography{bib}

\begin{thebibliography}{57}%
\makeatletter
\providecommand \@ifxundefined [1]{%
 \@ifx{#1\undefined}
}%
\providecommand \@ifnum [1]{%
 \ifnum #1\expandafter \@firstoftwo
 \else \expandafter \@secondoftwo
 \fi
}%
\providecommand \@ifx [1]{%
 \ifx #1\expandafter \@firstoftwo
 \else \expandafter \@secondoftwo
 \fi
}%
\providecommand \natexlab [1]{#1}%
\providecommand \enquote  [1]{``#1''}%
\providecommand \bibnamefont  [1]{#1}%
\providecommand \bibfnamefont [1]{#1}%
\providecommand \citenamefont [1]{#1}%
\providecommand \href@noop [0]{\@secondoftwo}%
\providecommand \href [0]{\begingroup \@sanitize@url \@href}%
\providecommand \@href[1]{\@@startlink{#1}\@@href}%
\providecommand \@@href[1]{\endgroup#1\@@endlink}%
\providecommand \@sanitize@url [0]{\catcode `\\12\catcode `\$12\catcode
  `\&12\catcode `\#12\catcode `\^12\catcode `\_12\catcode `\%12\relax}%
\providecommand \@@startlink[1]{}%
\providecommand \@@endlink[0]{}%
\providecommand \url  [0]{\begingroup\@sanitize@url \@url }%
\providecommand \@url [1]{\endgroup\@href {#1}{\urlprefix }}%
\providecommand \urlprefix  [0]{URL }%
\providecommand \Eprint [0]{\href }%
\providecommand \doibase [0]{http://dx.doi.org/}%
\providecommand \selectlanguage [0]{\@gobble}%
\providecommand \bibinfo  [0]{\@secondoftwo}%
\providecommand \bibfield  [0]{\@secondoftwo}%
\providecommand \translation [1]{[#1]}%
\providecommand \BibitemOpen [0]{}%
\providecommand \bibitemStop [0]{}%
\providecommand \bibitemNoStop [0]{.\EOS\space}%
\providecommand \EOS [0]{\spacefactor3000\relax}%
\providecommand \BibitemShut  [1]{\csname bibitem#1\endcsname}%
\let\auto@bib@innerbib\@empty
\bibitem [{\citenamefont {Braunstein}\ and\ \citenamefont
  {Caves}(1994)}]{Braunstein1994}%
  \BibitemOpen
  \bibfield  {author} {\bibinfo {author} {\bibfnamefont {S.~L.}\ \bibnamefont
  {Braunstein}}\ and\ \bibinfo {author} {\bibfnamefont {C.~M.}\ \bibnamefont
  {Caves}},\ }\href {\doibase 10.1103/PhysRevLett.72.3439} {\bibfield
  {journal} {\bibinfo  {journal} {Phys. Rev. Lett.}\ }\textbf {\bibinfo
  {volume} {72}},\ \bibinfo {pages} {3439} (\bibinfo {year}
  {1994})}\BibitemShut {NoStop}%
\bibitem [{\citenamefont {Giovannetti}\ \emph {et~al.}(2006)\citenamefont
  {Giovannetti}, \citenamefont {Lloyd},\ and\ \citenamefont
  {Maccone}}]{giovenneti2006}%
  \BibitemOpen
  \bibfield  {author} {\bibinfo {author} {\bibfnamefont {V.}~\bibnamefont
  {Giovannetti}}, \bibinfo {author} {\bibfnamefont {S.}~\bibnamefont {Lloyd}},
  \ and\ \bibinfo {author} {\bibfnamefont {L.}~\bibnamefont {Maccone}},\ }\href
  {\doibase 10.1103/PhysRevLett.96.010401} {\bibfield  {journal} {\bibinfo
  {journal} {Phys. Rev. Lett.}\ }\textbf {\bibinfo {volume} {96}},\ \bibinfo
  {pages} {010401} (\bibinfo {year} {2006})}\BibitemShut {NoStop}%
\bibitem [{\citenamefont {Giovannetti}\ \emph {et~al.}(2011)\citenamefont
  {Giovannetti}, \citenamefont {Lloyd},\ and\ \citenamefont
  {Maccone}}]{giovannetti_nature}%
  \BibitemOpen
  \bibfield  {author} {\bibinfo {author} {\bibfnamefont {V.}~\bibnamefont
  {Giovannetti}}, \bibinfo {author} {\bibfnamefont {S.}~\bibnamefont {Lloyd}},
  \ and\ \bibinfo {author} {\bibfnamefont {L.}~\bibnamefont {Maccone}},\ }\href
  {\doibase 10.1038/nphoton.2011.35} {\bibfield  {journal} {\bibinfo  {journal}
  {Nature Photonics}\ }\textbf {\bibinfo {volume} {5}},\ \bibinfo {pages} {222}
  (\bibinfo {year} {2011})}\BibitemShut {NoStop}%
\bibitem [{\citenamefont {Degen}\ \emph {et~al.}(2017)\citenamefont {Degen},
  \citenamefont {Reinhard},\ and\ \citenamefont
  {Cappellaro}}]{Sensing_RMP_2017}%
  \BibitemOpen
  \bibfield  {author} {\bibinfo {author} {\bibfnamefont {C.~L.}\ \bibnamefont
  {Degen}}, \bibinfo {author} {\bibfnamefont {F.}~\bibnamefont {Reinhard}}, \
  and\ \bibinfo {author} {\bibfnamefont {P.}~\bibnamefont {Cappellaro}},\
  }\href {\doibase 10.1103/RevModPhys.89.035002} {\bibfield  {journal}
  {\bibinfo  {journal} {Rev. Mod. Phys.}\ }\textbf {\bibinfo {volume} {89}},\
  \bibinfo {pages} {035002} (\bibinfo {year} {2017})}\BibitemShut {NoStop}%
\bibitem [{\citenamefont {Caves}(1981)}]{caves1981}%
  \BibitemOpen
  \bibfield  {author} {\bibinfo {author} {\bibfnamefont {C.~M.}\ \bibnamefont
  {Caves}},\ }\href {\doibase 10.1103/PhysRevD.23.1693} {\bibfield  {journal}
  {\bibinfo  {journal} {Phys. Rev. D}\ }\textbf {\bibinfo {volume} {23}},\
  \bibinfo {pages} {1693} (\bibinfo {year} {1981})}\BibitemShut {NoStop}%
\bibitem [{\citenamefont {Demkowicz-Dobrzański}\ \emph
  {et~al.}(2015)\citenamefont {Demkowicz-Dobrzański}, \citenamefont
  {Jarzyna},\ and\ \citenamefont {Kołodyński}}]{rafal2015}%
  \BibitemOpen
  \bibfield  {author} {\bibinfo {author} {\bibfnamefont {R.}~\bibnamefont
  {Demkowicz-Dobrzański}}, \bibinfo {author} {\bibfnamefont {M.}~\bibnamefont
  {Jarzyna}}, \ and\ \bibinfo {author} {\bibfnamefont {J.}~\bibnamefont
  {Kołodyński}}\ }(\bibinfo  {publisher} {Elsevier},\ \bibinfo {year}
  {2015})\ pp.\ \bibinfo {pages} {345--435}\BibitemShut {NoStop}%
\bibitem [{\citenamefont {Pirandola}\ \emph {et~al.}(2018)\citenamefont
  {Pirandola}, \citenamefont {Bardhan}, \citenamefont {Gehring}, \citenamefont
  {Weedbrook},\ and\ \citenamefont {Lloyd}}]{pirandola2018}%
  \BibitemOpen
  \bibfield  {author} {\bibinfo {author} {\bibfnamefont {S.}~\bibnamefont
  {Pirandola}}, \bibinfo {author} {\bibfnamefont {B.~R.}\ \bibnamefont
  {Bardhan}}, \bibinfo {author} {\bibfnamefont {T.}~\bibnamefont {Gehring}},
  \bibinfo {author} {\bibfnamefont {C.}~\bibnamefont {Weedbrook}}, \ and\
  \bibinfo {author} {\bibfnamefont {S.}~\bibnamefont {Lloyd}},\ }\href
  {https://doi.org/10.1038/s41566-018-0301-6} {\bibfield  {journal} {\bibinfo
  {journal} {Nature Photonics}\ }\textbf {\bibinfo {volume} {12}},\ \bibinfo
  {pages} {724} (\bibinfo {year} {2018})}\BibitemShut {NoStop}%
\bibitem [{\citenamefont {Candeloro}\ \emph {et~al.}(2020)\citenamefont
  {Candeloro}, \citenamefont {Boschi},\ and\ \citenamefont
  {Paris}}]{Alessandro2020}%
  \BibitemOpen
  \bibfield  {author} {\bibinfo {author} {\bibfnamefont {A.}~\bibnamefont
  {Candeloro}}, \bibinfo {author} {\bibfnamefont {C.~D.~E.}\ \bibnamefont
  {Boschi}}, \ and\ \bibinfo {author} {\bibfnamefont {M.~G.~A.}\ \bibnamefont
  {Paris}},\ }\href {\doibase 10.1103/PhysRevD.102.056012} {\bibfield
  {journal} {\bibinfo  {journal} {Phys. Rev. D}\ }\textbf {\bibinfo {volume}
  {102}},\ \bibinfo {pages} {056012} (\bibinfo {year} {2020})}\BibitemShut
  {NoStop}%
\bibitem [{\citenamefont {Albarelli}\ \emph {et~al.}(2020)\citenamefont
  {Albarelli}, \citenamefont {Barbieri}, \citenamefont {Genoni},\ and\
  \citenamefont {Gianani}}]{albarelli2020}%
  \BibitemOpen
  \bibfield  {author} {\bibinfo {author} {\bibfnamefont {F.}~\bibnamefont
  {Albarelli}}, \bibinfo {author} {\bibfnamefont {M.}~\bibnamefont {Barbieri}},
  \bibinfo {author} {\bibfnamefont {M.}~\bibnamefont {Genoni}}, \ and\ \bibinfo
  {author} {\bibfnamefont {I.}~\bibnamefont {Gianani}},\ }\href {\doibase
  https://doi.org/10.1016/j.physleta.2020.126311} {\bibfield  {journal}
  {\bibinfo  {journal} {Physics Letters A}\ }\textbf {\bibinfo {volume}
  {384}},\ \bibinfo {pages} {126311} (\bibinfo {year} {2020})}\BibitemShut
  {NoStop}%
\bibitem [{\citenamefont {Taylor}\ and\ \citenamefont
  {Bowen}(2016)}]{taylor2016}%
  \BibitemOpen
  \bibfield  {author} {\bibinfo {author} {\bibfnamefont {M.~A.}\ \bibnamefont
  {Taylor}}\ and\ \bibinfo {author} {\bibfnamefont {W.~P.}\ \bibnamefont
  {Bowen}},\ }\href {\doibase https://doi.org/10.1016/j.physrep.2015.12.002}
  {\bibfield  {journal} {\bibinfo  {journal} {Physics Reports}\ }\textbf
  {\bibinfo {volume} {615}},\ \bibinfo {pages} {1} (\bibinfo {year}
  {2016})}\BibitemShut {NoStop}%
\bibitem [{\citenamefont {Wootters}(1981)}]{wootters_1981}%
  \BibitemOpen
  \bibfield  {author} {\bibinfo {author} {\bibfnamefont {W.~K.}\ \bibnamefont
  {Wootters}},\ }\href {\doibase 10.1103/PhysRevD.23.357} {\bibfield  {journal}
  {\bibinfo  {journal} {Phys. Rev. D}\ }\textbf {\bibinfo {volume} {23}},\
  \bibinfo {pages} {357} (\bibinfo {year} {1981})}\BibitemShut {NoStop}%
\bibitem [{\citenamefont {Helstrom}(1976)}]{Helstrom1976}%
  \BibitemOpen
  \bibfield  {author} {\bibinfo {author} {\bibfnamefont {C.~W.}\ \bibnamefont
  {Helstrom}},\ }\href
  {https://www.sciencedirect.com/bookseries/mathematics-in-science-and-engineering/vol/123/suppl/C}
  {\emph {\bibinfo {title} {Quantum Detection and Estimation Theory}}}\
  (\bibinfo  {publisher} {Academic Press},\ \bibinfo {year} {1976})\BibitemShut
  {NoStop}%
\bibitem [{\citenamefont {Horodecki}\ \emph {et~al.}(2009)\citenamefont
  {Horodecki}, \citenamefont {Horodecki}, \citenamefont {Horodecki},\ and\
  \citenamefont {Horodecki}}]{HHHH_entanglement}%
  \BibitemOpen
  \bibfield  {author} {\bibinfo {author} {\bibfnamefont {R.}~\bibnamefont
  {Horodecki}}, \bibinfo {author} {\bibfnamefont {P.}~\bibnamefont
  {Horodecki}}, \bibinfo {author} {\bibfnamefont {M.}~\bibnamefont
  {Horodecki}}, \ and\ \bibinfo {author} {\bibfnamefont {K.}~\bibnamefont
  {Horodecki}},\ }\href {\doibase 10.1103/RevModPhys.81.865} {\bibfield
  {journal} {\bibinfo  {journal} {Rev. Mod. Phys.}\ }\textbf {\bibinfo {volume}
  {81}},\ \bibinfo {pages} {865} (\bibinfo {year} {2009})}\BibitemShut
  {NoStop}%
\bibitem [{\citenamefont {Pezz\'e}\ and\ \citenamefont
  {Smerzi}(2009)}]{luca_augusto_prl_2009}%
  \BibitemOpen
  \bibfield  {author} {\bibinfo {author} {\bibfnamefont {L.}~\bibnamefont
  {Pezz\'e}}\ and\ \bibinfo {author} {\bibfnamefont {A.}~\bibnamefont
  {Smerzi}},\ }\href {\doibase 10.1103/PhysRevLett.102.100401} {\bibfield
  {journal} {\bibinfo  {journal} {Phys. Rev. Lett.}\ }\textbf {\bibinfo
  {volume} {102}},\ \bibinfo {pages} {100401} (\bibinfo {year}
  {2009})}\BibitemShut {NoStop}%
\bibitem [{\citenamefont {Mishra}\ and\ \citenamefont
  {Bayat}(2021)}]{utkarsh2021}%
  \BibitemOpen
  \bibfield  {author} {\bibinfo {author} {\bibfnamefont {U.}~\bibnamefont
  {Mishra}}\ and\ \bibinfo {author} {\bibfnamefont {A.}~\bibnamefont {Bayat}},\
  }\href {\doibase 10.1103/PhysRevLett.127.080504} {\bibfield  {journal}
  {\bibinfo  {journal} {Phys. Rev. Lett.}\ }\textbf {\bibinfo {volume} {127}},\
  \bibinfo {pages} {080504} (\bibinfo {year} {2021})}\BibitemShut {NoStop}%
\bibitem [{\citenamefont {Garbe}\ \emph {et~al.}(2022)\citenamefont {Garbe},
  \citenamefont {Abah}, \citenamefont {Felicetti},\ and\ \citenamefont
  {Puebla}}]{louis2022}%
  \BibitemOpen
  \bibfield  {author} {\bibinfo {author} {\bibfnamefont {L.}~\bibnamefont
  {Garbe}}, \bibinfo {author} {\bibfnamefont {O.}~\bibnamefont {Abah}},
  \bibinfo {author} {\bibfnamefont {S.}~\bibnamefont {Felicetti}}, \ and\
  \bibinfo {author} {\bibfnamefont {R.}~\bibnamefont {Puebla}},\ }\href
  {\doibase 10.1103/PhysRevResearch.4.043061} {\bibfield  {journal} {\bibinfo
  {journal} {Phys. Rev. Res.}\ }\textbf {\bibinfo {volume} {4}},\ \bibinfo
  {pages} {043061} (\bibinfo {year} {2022})}\BibitemShut {NoStop}%
\bibitem [{\citenamefont {Schnabel}(2017)}]{schnabel2017}%
  \BibitemOpen
  \bibfield  {author} {\bibinfo {author} {\bibfnamefont {R.}~\bibnamefont
  {Schnabel}},\ }\href {\doibase https://doi.org/10.1016/j.physrep.2017.04.001}
  {\bibfield  {journal} {\bibinfo  {journal} {Physics Reports}\ }\textbf
  {\bibinfo {volume} {684}},\ \bibinfo {pages} {1} (\bibinfo {year} {2017})},\
  \bibinfo {note} {squeezed states of light and their applications in laser
  interferometers}\BibitemShut {NoStop}%
\bibitem [{\citenamefont {Zhao}\ \emph {et~al.}(2020)\citenamefont {Zhao},
  \citenamefont {Yang},\ and\ \citenamefont {Chiribella}}]{zhao2020}%
  \BibitemOpen
  \bibfield  {author} {\bibinfo {author} {\bibfnamefont {X.}~\bibnamefont
  {Zhao}}, \bibinfo {author} {\bibfnamefont {Y.}~\bibnamefont {Yang}}, \ and\
  \bibinfo {author} {\bibfnamefont {G.}~\bibnamefont {Chiribella}},\ }\href
  {\doibase 10.1103/PhysRevLett.124.190503} {\bibfield  {journal} {\bibinfo
  {journal} {Phys. Rev. Lett.}\ }\textbf {\bibinfo {volume} {124}},\ \bibinfo
  {pages} {190503} (\bibinfo {year} {2020})}\BibitemShut {NoStop}%
\bibitem [{\citenamefont {Sone}\ \emph {et~al.}(2019)\citenamefont {Sone},
  \citenamefont {Zhuang}, \citenamefont {Li}, \citenamefont {Liu},\ and\
  \citenamefont {Cappellaro}}]{sone2019}%
  \BibitemOpen
  \bibfield  {author} {\bibinfo {author} {\bibfnamefont {A.}~\bibnamefont
  {Sone}}, \bibinfo {author} {\bibfnamefont {Q.}~\bibnamefont {Zhuang}},
  \bibinfo {author} {\bibfnamefont {C.}~\bibnamefont {Li}}, \bibinfo {author}
  {\bibfnamefont {Y.-X.}\ \bibnamefont {Liu}}, \ and\ \bibinfo {author}
  {\bibfnamefont {P.}~\bibnamefont {Cappellaro}},\ }\href {\doibase
  10.1103/PhysRevA.99.052318} {\bibfield  {journal} {\bibinfo  {journal} {Phys.
  Rev. A}\ }\textbf {\bibinfo {volume} {99}},\ \bibinfo {pages} {052318}
  (\bibinfo {year} {2019})}\BibitemShut {NoStop}%
\bibitem [{\citenamefont {Yadin}\ \emph {et~al.}(2021)\citenamefont {Yadin},
  \citenamefont {Fadel},\ and\ \citenamefont {Gessner}}]{yadin2021}%
  \BibitemOpen
  \bibfield  {author} {\bibinfo {author} {\bibfnamefont {B.}~\bibnamefont
  {Yadin}}, \bibinfo {author} {\bibfnamefont {M.}~\bibnamefont {Fadel}}, \ and\
  \bibinfo {author} {\bibfnamefont {M.}~\bibnamefont {Gessner}},\ }\href
  {https://www.nature.com/articles/s41467-021-22353-3} {\bibfield  {journal}
  {\bibinfo  {journal} {Nature communications}\ }\textbf {\bibinfo {volume}
  {12}},\ \bibinfo {pages} {2410} (\bibinfo {year} {2021})}\BibitemShut
  {NoStop}%
\bibitem [{\citenamefont {Lee}\ \emph {et~al.}(2023)\citenamefont {Lee},
  \citenamefont {Lin}, \citenamefont {Miranowicz}, \citenamefont {Nori},
  \citenamefont {Ku},\ and\ \citenamefont {Chen}}]{lee2023}%
  \BibitemOpen
  \bibfield  {author} {\bibinfo {author} {\bibfnamefont {K.-Y.}\ \bibnamefont
  {Lee}}, \bibinfo {author} {\bibfnamefont {J.-D.}\ \bibnamefont {Lin}},
  \bibinfo {author} {\bibfnamefont {A.}~\bibnamefont {Miranowicz}}, \bibinfo
  {author} {\bibfnamefont {F.}~\bibnamefont {Nori}}, \bibinfo {author}
  {\bibfnamefont {H.-Y.}\ \bibnamefont {Ku}}, \ and\ \bibinfo {author}
  {\bibfnamefont {Y.-N.}\ \bibnamefont {Chen}},\ }\href {\doibase
  10.1103/PhysRevResearch.5.013103} {\bibfield  {journal} {\bibinfo  {journal}
  {Phys. Rev. Res.}\ }\textbf {\bibinfo {volume} {5}},\ \bibinfo {pages}
  {013103} (\bibinfo {year} {2023})}\BibitemShut {NoStop}%
\bibitem [{\citenamefont {Fr\'erot}\ and\ \citenamefont
  {Roscilde}(2018)}]{irenee2018}%
  \BibitemOpen
  \bibfield  {author} {\bibinfo {author} {\bibfnamefont {I.}~\bibnamefont
  {Fr\'erot}}\ and\ \bibinfo {author} {\bibfnamefont {T.}~\bibnamefont
  {Roscilde}},\ }\href {\doibase 10.1103/PhysRevLett.121.020402} {\bibfield
  {journal} {\bibinfo  {journal} {Phys. Rev. Lett.}\ }\textbf {\bibinfo
  {volume} {121}},\ \bibinfo {pages} {020402} (\bibinfo {year}
  {2018})}\BibitemShut {NoStop}%
\bibitem [{\citenamefont {Rams}\ \emph {et~al.}(2018)\citenamefont {Rams},
  \citenamefont {Sierant}, \citenamefont {Dutta}, \citenamefont {Horodecki},\
  and\ \citenamefont {Zakrzewski}}]{rams2018}%
  \BibitemOpen
  \bibfield  {author} {\bibinfo {author} {\bibfnamefont {M.~M.}\ \bibnamefont
  {Rams}}, \bibinfo {author} {\bibfnamefont {P.}~\bibnamefont {Sierant}},
  \bibinfo {author} {\bibfnamefont {O.}~\bibnamefont {Dutta}}, \bibinfo
  {author} {\bibfnamefont {P.}~\bibnamefont {Horodecki}}, \ and\ \bibinfo
  {author} {\bibfnamefont {J.}~\bibnamefont {Zakrzewski}},\ }\href {\doibase
  10.1103/PhysRevX.8.021022} {\bibfield  {journal} {\bibinfo  {journal} {Phys.
  Rev. X}\ }\textbf {\bibinfo {volume} {8}},\ \bibinfo {pages} {021022}
  (\bibinfo {year} {2018})}\BibitemShut {NoStop}%
\bibitem [{\citenamefont {Montenegro}\ \emph {et~al.}(2022)\citenamefont
  {Montenegro}, \citenamefont {Jones}, \citenamefont {Bose},\ and\
  \citenamefont {Bayat}}]{victor2022}%
  \BibitemOpen
  \bibfield  {author} {\bibinfo {author} {\bibfnamefont {V.}~\bibnamefont
  {Montenegro}}, \bibinfo {author} {\bibfnamefont {G.~S.}\ \bibnamefont
  {Jones}}, \bibinfo {author} {\bibfnamefont {S.}~\bibnamefont {Bose}}, \ and\
  \bibinfo {author} {\bibfnamefont {A.}~\bibnamefont {Bayat}},\ }\href
  {\doibase 10.1103/PhysRevLett.129.120503} {\bibfield  {journal} {\bibinfo
  {journal} {Phys. Rev. Lett.}\ }\textbf {\bibinfo {volume} {129}},\ \bibinfo
  {pages} {120503} (\bibinfo {year} {2022})}\BibitemShut {NoStop}%
\bibitem [{\citenamefont {Roy}\ and\ \citenamefont
  {Braunstein}(2008)}]{roy2008}%
  \BibitemOpen
  \bibfield  {author} {\bibinfo {author} {\bibfnamefont {S.~M.}\ \bibnamefont
  {Roy}}\ and\ \bibinfo {author} {\bibfnamefont {S.~L.}\ \bibnamefont
  {Braunstein}},\ }\href {\doibase 10.1103/PhysRevLett.100.220501} {\bibfield
  {journal} {\bibinfo  {journal} {Phys. Rev. Lett.}\ }\textbf {\bibinfo
  {volume} {100}},\ \bibinfo {pages} {220501} (\bibinfo {year}
  {2008})}\BibitemShut {NoStop}%
\bibitem [{\citenamefont {Gietka}\ \emph {et~al.}(2022)\citenamefont {Gietka},
  \citenamefont {Ruks},\ and\ \citenamefont {Busch}}]{Gietka2022}%
  \BibitemOpen
  \bibfield  {author} {\bibinfo {author} {\bibfnamefont {K.}~\bibnamefont
  {Gietka}}, \bibinfo {author} {\bibfnamefont {L.}~\bibnamefont {Ruks}}, \ and\
  \bibinfo {author} {\bibfnamefont {T.}~\bibnamefont {Busch}},\ }\href
  {\doibase 10.22331/q-2022-04-27-700} {\bibfield  {journal} {\bibinfo
  {journal} {{Quantum}}\ }\textbf {\bibinfo {volume} {6}},\ \bibinfo {pages}
  {700} (\bibinfo {year} {2022})}\BibitemShut {NoStop}%
\bibitem [{\citenamefont {Boixo}\ \emph {et~al.}(2007)\citenamefont {Boixo},
  \citenamefont {Flammia}, \citenamefont {Caves},\ and\ \citenamefont
  {Geremia}}]{sergio2007}%
  \BibitemOpen
  \bibfield  {author} {\bibinfo {author} {\bibfnamefont {S.}~\bibnamefont
  {Boixo}}, \bibinfo {author} {\bibfnamefont {S.~T.}\ \bibnamefont {Flammia}},
  \bibinfo {author} {\bibfnamefont {C.~M.}\ \bibnamefont {Caves}}, \ and\
  \bibinfo {author} {\bibfnamefont {J.}~\bibnamefont {Geremia}},\ }\href
  {\doibase 10.1103/PhysRevLett.98.090401} {\bibfield  {journal} {\bibinfo
  {journal} {Phys. Rev. Lett.}\ }\textbf {\bibinfo {volume} {98}},\ \bibinfo
  {pages} {090401} (\bibinfo {year} {2007})}\BibitemShut {NoStop}%
\bibitem [{\citenamefont {Dooley}\ \emph {et~al.}(2023)\citenamefont {Dooley},
  \citenamefont {Pappalardi},\ and\ \citenamefont {Goold}}]{shane2023}%
  \BibitemOpen
  \bibfield  {author} {\bibinfo {author} {\bibfnamefont {S.}~\bibnamefont
  {Dooley}}, \bibinfo {author} {\bibfnamefont {S.}~\bibnamefont {Pappalardi}},
  \ and\ \bibinfo {author} {\bibfnamefont {J.}~\bibnamefont {Goold}},\ }\href
  {\doibase 10.1103/PhysRevB.107.035123} {\bibfield  {journal} {\bibinfo
  {journal} {Phys. Rev. B}\ }\textbf {\bibinfo {volume} {107}},\ \bibinfo
  {pages} {035123} (\bibinfo {year} {2023})}\BibitemShut {NoStop}%
\bibitem [{\citenamefont {He}\ \emph {et~al.}(2023)\citenamefont {He},
  \citenamefont {Yousefjani},\ and\ \citenamefont {Bayat}}]{Xingjian2023}%
  \BibitemOpen
  \bibfield  {author} {\bibinfo {author} {\bibfnamefont {X.}~\bibnamefont
  {He}}, \bibinfo {author} {\bibfnamefont {R.}~\bibnamefont {Yousefjani}}, \
  and\ \bibinfo {author} {\bibfnamefont {A.}~\bibnamefont {Bayat}},\ }\href
  {\doibase 10.1103/PhysRevLett.131.010801} {\bibfield  {journal} {\bibinfo
  {journal} {Phys. Rev. Lett.}\ }\textbf {\bibinfo {volume} {131}},\ \bibinfo
  {pages} {010801} (\bibinfo {year} {2023})}\BibitemShut {NoStop}%
\bibitem [{\citenamefont {Sahoo}\ \emph {et~al.}(2023)\citenamefont {Sahoo},
  \citenamefont {Mishra},\ and\ \citenamefont {Rakshit}}]{sahoo2023}%
  \BibitemOpen
  \bibfield  {author} {\bibinfo {author} {\bibfnamefont {A.}~\bibnamefont
  {Sahoo}}, \bibinfo {author} {\bibfnamefont {U.}~\bibnamefont {Mishra}}, \
  and\ \bibinfo {author} {\bibfnamefont {D.}~\bibnamefont {Rakshit}},\
  }\href@noop {} {\enquote {\bibinfo {title} {Localization driven quantum
  sensing},}\ } (\bibinfo {year} {2023}),\ \Eprint
  {http://arxiv.org/abs/2305.02315} {arXiv:2305.02315 [quant-ph]} \BibitemShut
  {NoStop}%
\bibitem [{\citenamefont {Shlyakhov}\ \emph {et~al.}(2018)\citenamefont
  {Shlyakhov}, \citenamefont {Zemlyanov}, \citenamefont {Suslov}, \citenamefont
  {Lebedev}, \citenamefont {Paraoanu}, \citenamefont {Lesovik},\ and\
  \citenamefont {Blatter}}]{Shlyakhov2018}%
  \BibitemOpen
  \bibfield  {author} {\bibinfo {author} {\bibfnamefont {A.~R.}\ \bibnamefont
  {Shlyakhov}}, \bibinfo {author} {\bibfnamefont {V.~V.}\ \bibnamefont
  {Zemlyanov}}, \bibinfo {author} {\bibfnamefont {M.~V.}\ \bibnamefont
  {Suslov}}, \bibinfo {author} {\bibfnamefont {A.~V.}\ \bibnamefont {Lebedev}},
  \bibinfo {author} {\bibfnamefont {G.~S.}\ \bibnamefont {Paraoanu}}, \bibinfo
  {author} {\bibfnamefont {G.~B.}\ \bibnamefont {Lesovik}}, \ and\ \bibinfo
  {author} {\bibfnamefont {G.}~\bibnamefont {Blatter}},\ }\href {\doibase
  10.1103/PhysRevA.97.022115} {\bibfield  {journal} {\bibinfo  {journal} {Phys.
  Rev. A}\ }\textbf {\bibinfo {volume} {97}},\ \bibinfo {pages} {022115}
  (\bibinfo {year} {2018})}\BibitemShut {NoStop}%
\bibitem [{\citenamefont {Dooley}(2021)}]{shane2021}%
  \BibitemOpen
  \bibfield  {author} {\bibinfo {author} {\bibfnamefont {S.}~\bibnamefont
  {Dooley}},\ }\href {\doibase 10.1103/PRXQuantum.2.020330} {\bibfield
  {journal} {\bibinfo  {journal} {PRX Quantum}\ }\textbf {\bibinfo {volume}
  {2}},\ \bibinfo {pages} {020330} (\bibinfo {year} {2021})}\BibitemShut
  {NoStop}%
\bibitem [{\citenamefont {Bartlett}\ \emph {et~al.}(2002)\citenamefont
  {Bartlett}, \citenamefont {de~Guise},\ and\ \citenamefont
  {Sanders}}]{sanders02}%
  \BibitemOpen
  \bibfield  {author} {\bibinfo {author} {\bibfnamefont {S.~D.}\ \bibnamefont
  {Bartlett}}, \bibinfo {author} {\bibfnamefont {H.}~\bibnamefont {de~Guise}},
  \ and\ \bibinfo {author} {\bibfnamefont {B.~C.}\ \bibnamefont {Sanders}},\
  }\href {\doibase 10.1103/PhysRevA.65.052316} {\bibfield  {journal} {\bibinfo
  {journal} {Phys. Rev. A}\ }\textbf {\bibinfo {volume} {65}},\ \bibinfo
  {pages} {052316} (\bibinfo {year} {2002})}\BibitemShut {NoStop}%
\bibitem [{\citenamefont {Cozzolino}\ \emph {et~al.}(2019)\citenamefont
  {Cozzolino}, \citenamefont {Da~Lio}, \citenamefont {Bacco},\ and\
  \citenamefont {Oxenløwe}}]{quditQtech}%
  \BibitemOpen
  \bibfield  {author} {\bibinfo {author} {\bibfnamefont {D.}~\bibnamefont
  {Cozzolino}}, \bibinfo {author} {\bibfnamefont {B.}~\bibnamefont {Da~Lio}},
  \bibinfo {author} {\bibfnamefont {D.}~\bibnamefont {Bacco}}, \ and\ \bibinfo
  {author} {\bibfnamefont {L.~K.}\ \bibnamefont {Oxenløwe}},\ }\href {\doibase
  https://doi.org/10.1002/qute.201900038} {\bibfield  {journal} {\bibinfo
  {journal} {Advanced Quantum Technologies}\ }\textbf {\bibinfo {volume} {2}},\
  \bibinfo {pages} {1900038} (\bibinfo {year} {2019})}\BibitemShut {NoStop}%
\bibitem [{\citenamefont {Wang}\ \emph {et~al.}(2020)\citenamefont {Wang},
  \citenamefont {Hu}, \citenamefont {Sanders},\ and\ \citenamefont
  {Kais}}]{wang2020}%
  \BibitemOpen
  \bibfield  {author} {\bibinfo {author} {\bibfnamefont {Y.}~\bibnamefont
  {Wang}}, \bibinfo {author} {\bibfnamefont {Z.}~\bibnamefont {Hu}}, \bibinfo
  {author} {\bibfnamefont {B.~C.}\ \bibnamefont {Sanders}}, \ and\ \bibinfo
  {author} {\bibfnamefont {S.}~\bibnamefont {Kais}},\ }\href
  {https://www.frontiersin.org/articles/10.3389/fphy.2020.589504/full}
  {\bibfield  {journal} {\bibinfo  {journal} {Frontiers in Physics}\ }\textbf
  {\bibinfo {volume} {8}},\ \bibinfo {pages} {589504} (\bibinfo {year}
  {2020})}\BibitemShut {NoStop}%
\bibitem [{\citenamefont {Ghosh}\ and\ \citenamefont
  {Sen(De)}(2022)}]{srijon2022}%
  \BibitemOpen
  \bibfield  {author} {\bibinfo {author} {\bibfnamefont {S.}~\bibnamefont
  {Ghosh}}\ and\ \bibinfo {author} {\bibfnamefont {A.}~\bibnamefont
  {Sen(De)}},\ }\href {\doibase 10.1103/PhysRevA.105.022628} {\bibfield
  {journal} {\bibinfo  {journal} {Phys. Rev. A}\ }\textbf {\bibinfo {volume}
  {105}},\ \bibinfo {pages} {022628} (\bibinfo {year} {2022})}\BibitemShut
  {NoStop}%
\bibitem [{\citenamefont {Konar}\ \emph {et~al.}(2023)\citenamefont {Konar},
  \citenamefont {Ghosh}, \citenamefont {Pal},\ and\ \citenamefont
  {Sen(De)}}]{tanoy2023}%
  \BibitemOpen
  \bibfield  {author} {\bibinfo {author} {\bibfnamefont {T.~K.}\ \bibnamefont
  {Konar}}, \bibinfo {author} {\bibfnamefont {S.}~\bibnamefont {Ghosh}},
  \bibinfo {author} {\bibfnamefont {A.~K.}\ \bibnamefont {Pal}}, \ and\
  \bibinfo {author} {\bibfnamefont {A.}~\bibnamefont {Sen(De)}},\ }\href
  {\doibase 10.1103/PhysRevA.107.032602} {\bibfield  {journal} {\bibinfo
  {journal} {Phys. Rev. A}\ }\textbf {\bibinfo {volume} {107}},\ \bibinfo
  {pages} {032602} (\bibinfo {year} {2023})}\BibitemShut {NoStop}%
\bibitem [{\citenamefont {Erhard}\ \emph {et~al.}(2018)\citenamefont {Erhard},
  \citenamefont {Fickler}, \citenamefont {Krenn},\ and\ \citenamefont
  {Zeilinger}}]{photonexp}%
  \BibitemOpen
  \bibfield  {author} {\bibinfo {author} {\bibfnamefont {M.}~\bibnamefont
  {Erhard}}, \bibinfo {author} {\bibfnamefont {R.}~\bibnamefont {Fickler}},
  \bibinfo {author} {\bibfnamefont {M.}~\bibnamefont {Krenn}}, \ and\ \bibinfo
  {author} {\bibfnamefont {A.}~\bibnamefont {Zeilinger}},\ }\href
  {https://doi.org/10.1038/lsa.2017.146} {\bibfield  {journal} {\bibinfo
  {journal} {Light: Science \& Applications}\ }\textbf {\bibinfo {volume}
  {7}},\ \bibinfo {pages} {17146} (\bibinfo {year} {2018})}\BibitemShut
  {NoStop}%
\bibitem [{\citenamefont {Low}\ \emph {et~al.}(2020)\citenamefont {Low},
  \citenamefont {White}, \citenamefont {Cox}, \citenamefont {Day},\ and\
  \citenamefont {Senko}}]{ionqudits}%
  \BibitemOpen
  \bibfield  {author} {\bibinfo {author} {\bibfnamefont {P.~J.}\ \bibnamefont
  {Low}}, \bibinfo {author} {\bibfnamefont {B.~M.}\ \bibnamefont {White}},
  \bibinfo {author} {\bibfnamefont {A.~A.}\ \bibnamefont {Cox}}, \bibinfo
  {author} {\bibfnamefont {M.~L.}\ \bibnamefont {Day}}, \ and\ \bibinfo
  {author} {\bibfnamefont {C.}~\bibnamefont {Senko}},\ }\href {\doibase
  10.1103/PhysRevResearch.2.033128} {\bibfield  {journal} {\bibinfo  {journal}
  {Phys. Rev. Research}\ }\textbf {\bibinfo {volume} {2}},\ \bibinfo {pages}
  {033128} (\bibinfo {year} {2020})}\BibitemShut {NoStop}%
\bibitem [{\citenamefont {Soltamov}\ \emph {et~al.}(2019)\citenamefont
  {Soltamov}, \citenamefont {Kasper}, \citenamefont {Poshakinskiy},
  \citenamefont {Anisimov}, \citenamefont {Mokhov}, \citenamefont {Sperlich},
  \citenamefont {Tarasenko}, \citenamefont {Baranov}, \citenamefont
  {Astakhov},\ and\ \citenamefont {Dyakonov}}]{NVcentrequdits}%
  \BibitemOpen
  \bibfield  {author} {\bibinfo {author} {\bibfnamefont {V.~A.}\ \bibnamefont
  {Soltamov}}, \bibinfo {author} {\bibfnamefont {C.}~\bibnamefont {Kasper}},
  \bibinfo {author} {\bibfnamefont {A.~V.}\ \bibnamefont {Poshakinskiy}},
  \bibinfo {author} {\bibfnamefont {A.~N.}\ \bibnamefont {Anisimov}}, \bibinfo
  {author} {\bibfnamefont {E.~N.}\ \bibnamefont {Mokhov}}, \bibinfo {author}
  {\bibfnamefont {A.}~\bibnamefont {Sperlich}}, \bibinfo {author}
  {\bibfnamefont {S.~A.}\ \bibnamefont {Tarasenko}}, \bibinfo {author}
  {\bibfnamefont {P.~G.}\ \bibnamefont {Baranov}}, \bibinfo {author}
  {\bibfnamefont {G.~V.}\ \bibnamefont {Astakhov}}, \ and\ \bibinfo {author}
  {\bibfnamefont {V.}~\bibnamefont {Dyakonov}},\ }\href@noop {} {\bibfield
  {journal} {\bibinfo  {journal} {Nature Communications}\ }\textbf {\bibinfo
  {volume} {10}},\ \bibinfo {pages} {1678} (\bibinfo {year}
  {2019})}\BibitemShut {NoStop}%
\bibitem [{\citenamefont {Neeley}\ \emph {et~al.}(2009)\citenamefont {Neeley},
  \citenamefont {Ansmann}, \citenamefont {Bialczak}, \citenamefont {Hofheinz},
  \citenamefont {Lucero}, \citenamefont {O'Connell}, \citenamefont {Sank},
  \citenamefont {Wang}, \citenamefont {Wenner}, \citenamefont {Cleland},
  \citenamefont {Geller},\ and\ \citenamefont {Martinis}}]{supercondqudits}%
  \BibitemOpen
  \bibfield  {author} {\bibinfo {author} {\bibfnamefont {M.}~\bibnamefont
  {Neeley}}, \bibinfo {author} {\bibfnamefont {M.}~\bibnamefont {Ansmann}},
  \bibinfo {author} {\bibfnamefont {R.~C.}\ \bibnamefont {Bialczak}}, \bibinfo
  {author} {\bibfnamefont {M.}~\bibnamefont {Hofheinz}}, \bibinfo {author}
  {\bibfnamefont {E.}~\bibnamefont {Lucero}}, \bibinfo {author} {\bibfnamefont
  {A.~D.}\ \bibnamefont {O'Connell}}, \bibinfo {author} {\bibfnamefont
  {D.}~\bibnamefont {Sank}}, \bibinfo {author} {\bibfnamefont {H.}~\bibnamefont
  {Wang}}, \bibinfo {author} {\bibfnamefont {J.}~\bibnamefont {Wenner}},
  \bibinfo {author} {\bibfnamefont {A.~N.}\ \bibnamefont {Cleland}}, \bibinfo
  {author} {\bibfnamefont {M.~R.}\ \bibnamefont {Geller}}, \ and\ \bibinfo
  {author} {\bibfnamefont {J.~M.}\ \bibnamefont {Martinis}},\ }\href {\doibase
  10.1126/science.1173440} {\bibfield  {journal} {\bibinfo  {journal}
  {Science}\ }\textbf {\bibinfo {volume} {325}},\ \bibinfo {pages} {722}
  (\bibinfo {year} {2009})}\BibitemShut {NoStop}%
\bibitem [{\citenamefont {T\'oth}\ and\ \citenamefont
  {V\'ertesi}(2018{\natexlab{a}})}]{GToth2018}%
  \BibitemOpen
  \bibfield  {author} {\bibinfo {author} {\bibfnamefont {G.}~\bibnamefont
  {T\'oth}}\ and\ \bibinfo {author} {\bibfnamefont {T.}~\bibnamefont
  {V\'ertesi}},\ }\href {\doibase 10.1103/PhysRevLett.120.020506} {\bibfield
  {journal} {\bibinfo  {journal} {Phys. Rev. Lett.}\ }\textbf {\bibinfo
  {volume} {120}},\ \bibinfo {pages} {020506} (\bibinfo {year}
  {2018}{\natexlab{a}})}\BibitemShut {NoStop}%
\bibitem [{\citenamefont {PARIS}(2009)}]{paris2009}%
  \BibitemOpen
  \bibfield  {author} {\bibinfo {author} {\bibfnamefont {M.~G.~A.}\
  \bibnamefont {PARIS}},\ }\href {\doibase 10.1142/S0219749909004839}
  {\bibfield  {journal} {\bibinfo  {journal} {International Journal of Quantum
  Information}\ }\textbf {\bibinfo {volume} {07}},\ \bibinfo {pages} {125}
  (\bibinfo {year} {2009})},\ \Eprint
  {http://arxiv.org/abs/https://doi.org/10.1142/S0219749909004839}
  {https://doi.org/10.1142/S0219749909004839} \BibitemShut {NoStop}%
\bibitem [{\citenamefont {T\'oth}\ and\ \citenamefont
  {V\'ertesi}(2018{\natexlab{b}})}]{geza_prl_bound_2018}%
  \BibitemOpen
  \bibfield  {author} {\bibinfo {author} {\bibfnamefont {G.}~\bibnamefont
  {T\'oth}}\ and\ \bibinfo {author} {\bibfnamefont {T.}~\bibnamefont
  {V\'ertesi}},\ }\href {\doibase 10.1103/PhysRevLett.120.020506} {\bibfield
  {journal} {\bibinfo  {journal} {Phys. Rev. Lett.}\ }\textbf {\bibinfo
  {volume} {120}},\ \bibinfo {pages} {020506} (\bibinfo {year}
  {2018}{\natexlab{b}})}\BibitemShut {NoStop}%
\bibitem [{\citenamefont {T\'oth}(2012)}]{geza2012}%
  \BibitemOpen
  \bibfield  {author} {\bibinfo {author} {\bibfnamefont {G.}~\bibnamefont
  {T\'oth}},\ }\href {\doibase 10.1103/PhysRevA.85.022322} {\bibfield
  {journal} {\bibinfo  {journal} {Phys. Rev. A}\ }\textbf {\bibinfo {volume}
  {85}},\ \bibinfo {pages} {022322} (\bibinfo {year} {2012})}\BibitemShut
  {NoStop}%
\bibitem [{\citenamefont {Huelga}\ \emph {et~al.}(1997)\citenamefont {Huelga},
  \citenamefont {Macchiavello}, \citenamefont {Pellizzari}, \citenamefont
  {Ekert}, \citenamefont {Plenio},\ and\ \citenamefont {Cirac}}]{huelga}%
  \BibitemOpen
  \bibfield  {author} {\bibinfo {author} {\bibfnamefont {S.~F.}\ \bibnamefont
  {Huelga}}, \bibinfo {author} {\bibfnamefont {C.}~\bibnamefont
  {Macchiavello}}, \bibinfo {author} {\bibfnamefont {T.}~\bibnamefont
  {Pellizzari}}, \bibinfo {author} {\bibfnamefont {A.~K.}\ \bibnamefont
  {Ekert}}, \bibinfo {author} {\bibfnamefont {M.~B.}\ \bibnamefont {Plenio}}, \
  and\ \bibinfo {author} {\bibfnamefont {J.~I.}\ \bibnamefont {Cirac}},\ }\href
  {\doibase 10.1103/PhysRevLett.79.3865} {\bibfield  {journal} {\bibinfo
  {journal} {Phys. Rev. Lett.}\ }\textbf {\bibinfo {volume} {79}},\ \bibinfo
  {pages} {3865} (\bibinfo {year} {1997})}\BibitemShut {NoStop}%
\bibitem [{\citenamefont {Matsuzaki}\ \emph {et~al.}(2011)\citenamefont
  {Matsuzaki}, \citenamefont {Benjamin},\ and\ \citenamefont
  {Fitzsimons}}]{matsuzaki2011}%
  \BibitemOpen
  \bibfield  {author} {\bibinfo {author} {\bibfnamefont {Y.}~\bibnamefont
  {Matsuzaki}}, \bibinfo {author} {\bibfnamefont {S.~C.}\ \bibnamefont
  {Benjamin}}, \ and\ \bibinfo {author} {\bibfnamefont {J.}~\bibnamefont
  {Fitzsimons}},\ }\href {\doibase 10.1103/PhysRevA.84.012103} {\bibfield
  {journal} {\bibinfo  {journal} {Phys. Rev. A}\ }\textbf {\bibinfo {volume}
  {84}},\ \bibinfo {pages} {012103} (\bibinfo {year} {2011})}\BibitemShut
  {NoStop}%
\bibitem [{\citenamefont {Yoshinaga}\ \emph {et~al.}(2021)\citenamefont
  {Yoshinaga}, \citenamefont {Tatsuta},\ and\ \citenamefont
  {Matsuzaki}}]{Matsuzaki_2021}%
  \BibitemOpen
  \bibfield  {author} {\bibinfo {author} {\bibfnamefont {A.}~\bibnamefont
  {Yoshinaga}}, \bibinfo {author} {\bibfnamefont {M.}~\bibnamefont {Tatsuta}},
  \ and\ \bibinfo {author} {\bibfnamefont {Y.}~\bibnamefont {Matsuzaki}},\
  }\href {\doibase 10.1103/PhysRevA.103.062602} {\bibfield  {journal} {\bibinfo
   {journal} {Phys. Rev. A}\ }\textbf {\bibinfo {volume} {103}},\ \bibinfo
  {pages} {062602} (\bibinfo {year} {2021})}\BibitemShut {NoStop}%
\bibitem [{\citenamefont {Bengtsson}\ and\ \citenamefont
  {Zyczkowski}(2006)}]{Bengtsson_Zyczkowski_2006}%
  \BibitemOpen
  \bibfield  {author} {\bibinfo {author} {\bibfnamefont {I.}~\bibnamefont
  {Bengtsson}}\ and\ \bibinfo {author} {\bibfnamefont {K.}~\bibnamefont
  {Zyczkowski}},\ }\href@noop {} {\emph {\bibinfo {title} {Geometry of Quantum
  States: An Introduction to Quantum Entanglement}}}\ (\bibinfo  {publisher}
  {Cambridge University Press},\ \bibinfo {year} {2006})\BibitemShut {NoStop}%
\bibitem [{\citenamefont {Monika}\ \emph {et~al.}(2023)\citenamefont {Monika},
  \citenamefont {Lakkaraju}, \citenamefont {Ghosh},\ and\ \citenamefont
  {De}}]{monika2023}%
  \BibitemOpen
  \bibfield  {author} {\bibinfo {author} {\bibnamefont {Monika}}, \bibinfo
  {author} {\bibfnamefont {L.~G.~C.}\ \bibnamefont {Lakkaraju}}, \bibinfo
  {author} {\bibfnamefont {S.}~\bibnamefont {Ghosh}}, \ and\ \bibinfo {author}
  {\bibfnamefont {A.~S.}\ \bibnamefont {De}},\ }\href@noop {} {\enquote
  {\bibinfo {title} {Better sensing with variable-range interactions},}\ }
  (\bibinfo {year} {2023}),\ \Eprint {http://arxiv.org/abs/2307.06901}
  {arXiv:2307.06901 [quant-ph]} \BibitemShut {NoStop}%
\bibitem [{\citenamefont {Yousefjani}\ \emph {et~al.}(2023)\citenamefont
  {Yousefjani}, \citenamefont {He},\ and\ \citenamefont
  {Bayat}}]{abolfazl_longrange_2023}%
  \BibitemOpen
  \bibfield  {author} {\bibinfo {author} {\bibfnamefont {R.}~\bibnamefont
  {Yousefjani}}, \bibinfo {author} {\bibfnamefont {X.}~\bibnamefont {He}}, \
  and\ \bibinfo {author} {\bibfnamefont {A.}~\bibnamefont {Bayat}},\ }\href
  {\doibase 10.1088/1674-1056/acf302} {\bibfield  {journal} {\bibinfo
  {journal} {Chinese Physics B}\ }\textbf {\bibinfo {volume} {32}},\ \bibinfo
  {pages} {100313} (\bibinfo {year} {2023})}\BibitemShut {NoStop}%
\bibitem [{\citenamefont {Sadhukhan}\ \emph {et~al.}(2020)\citenamefont
  {Sadhukhan}, \citenamefont {Sinha}, \citenamefont {Francuz}, \citenamefont
  {Stefaniak}, \citenamefont {Rams}, \citenamefont {Dziarmaga},\ and\
  \citenamefont {Zurek}}]{debasis_sadhukhan_2020}%
  \BibitemOpen
  \bibfield  {author} {\bibinfo {author} {\bibfnamefont {D.}~\bibnamefont
  {Sadhukhan}}, \bibinfo {author} {\bibfnamefont {A.}~\bibnamefont {Sinha}},
  \bibinfo {author} {\bibfnamefont {A.}~\bibnamefont {Francuz}}, \bibinfo
  {author} {\bibfnamefont {J.}~\bibnamefont {Stefaniak}}, \bibinfo {author}
  {\bibfnamefont {M.~M.}\ \bibnamefont {Rams}}, \bibinfo {author}
  {\bibfnamefont {J.}~\bibnamefont {Dziarmaga}}, \ and\ \bibinfo {author}
  {\bibfnamefont {W.~H.}\ \bibnamefont {Zurek}},\ }\href {\doibase
  10.1103/PhysRevB.101.144429} {\bibfield  {journal} {\bibinfo  {journal}
  {Phys. Rev. B}\ }\textbf {\bibinfo {volume} {101}},\ \bibinfo {pages}
  {144429} (\bibinfo {year} {2020})}\BibitemShut {NoStop}%
\bibitem [{\citenamefont {Lakkaraju}\ \emph {et~al.}(2022)\citenamefont
  {Lakkaraju}, \citenamefont {Ghosh}, \citenamefont {Sadhukhan},\ and\
  \citenamefont {Sen(De)}}]{mimicking_lakkaraju_2022}%
  \BibitemOpen
  \bibfield  {author} {\bibinfo {author} {\bibfnamefont {L.~G.~C.}\
  \bibnamefont {Lakkaraju}}, \bibinfo {author} {\bibfnamefont {S.}~\bibnamefont
  {Ghosh}}, \bibinfo {author} {\bibfnamefont {D.}~\bibnamefont {Sadhukhan}}, \
  and\ \bibinfo {author} {\bibfnamefont {A.}~\bibnamefont {Sen(De)}},\ }\href
  {\doibase 10.1103/PhysRevA.106.052425} {\bibfield  {journal} {\bibinfo
  {journal} {Phys. Rev. A}\ }\textbf {\bibinfo {volume} {106}},\ \bibinfo
  {pages} {052425} (\bibinfo {year} {2022})}\BibitemShut {NoStop}%
\bibitem [{\citenamefont {Lakkaraju}\ \emph {et~al.}(2023)\citenamefont
  {Lakkaraju}, \citenamefont {Ghosh}, \citenamefont {Sadhukhan},\ and\
  \citenamefont {De}}]{dqpt_lakkaraju_2023}%
  \BibitemOpen
  \bibfield  {author} {\bibinfo {author} {\bibfnamefont {L.~G.~C.}\
  \bibnamefont {Lakkaraju}}, \bibinfo {author} {\bibfnamefont {S.}~\bibnamefont
  {Ghosh}}, \bibinfo {author} {\bibfnamefont {D.}~\bibnamefont {Sadhukhan}}, \
  and\ \bibinfo {author} {\bibfnamefont {A.~S.}\ \bibnamefont {De}},\
  }\href@noop {} {\enquote {\bibinfo {title} {Framework of dynamical
  transitions from long-range to short-range quantum systems},}\ } (\bibinfo
  {year} {2023}),\ \Eprint {http://arxiv.org/abs/2305.02945} {arXiv:2305.02945
  [quant-ph]} \BibitemShut {NoStop}%
\bibitem [{\citenamefont {Vodola}\ \emph {et~al.}(2015)\citenamefont {Vodola},
  \citenamefont {Lepori}, \citenamefont {Ercolessi},\ and\ \citenamefont
  {Pupillo}}]{vodola2015longrangephases}%
  \BibitemOpen
  \bibfield  {author} {\bibinfo {author} {\bibfnamefont {D.}~\bibnamefont
  {Vodola}}, \bibinfo {author} {\bibfnamefont {L.}~\bibnamefont {Lepori}},
  \bibinfo {author} {\bibfnamefont {E.}~\bibnamefont {Ercolessi}}, \ and\
  \bibinfo {author} {\bibfnamefont {G.}~\bibnamefont {Pupillo}},\ }\href
  {https://iopscience.iop.org/article/10.1088/1367-2630/18/1/015001} {\bibfield
   {journal} {\bibinfo  {journal} {New Journal of Physics}\ }\textbf {\bibinfo
  {volume} {18}},\ \bibinfo {pages} {015001} (\bibinfo {year}
  {2015})}\BibitemShut {NoStop}%
\bibitem [{\citenamefont {Vodola}\ \emph {et~al.}(2014)\citenamefont {Vodola},
  \citenamefont {Lepori}, \citenamefont {Ercolessi}, \citenamefont {Gorshkov},\
  and\ \citenamefont {Pupillo}}]{Vodola1}%
  \BibitemOpen
  \bibfield  {author} {\bibinfo {author} {\bibfnamefont {D.}~\bibnamefont
  {Vodola}}, \bibinfo {author} {\bibfnamefont {L.}~\bibnamefont {Lepori}},
  \bibinfo {author} {\bibfnamefont {E.}~\bibnamefont {Ercolessi}}, \bibinfo
  {author} {\bibfnamefont {A.~V.}\ \bibnamefont {Gorshkov}}, \ and\ \bibinfo
  {author} {\bibfnamefont {G.}~\bibnamefont {Pupillo}},\ }\href {\doibase
  10.1103/PhysRevLett.113.156402} {\bibfield  {journal} {\bibinfo  {journal}
  {Phys. Rev. Lett.}\ }\textbf {\bibinfo {volume} {113}},\ \bibinfo {pages}
  {156402} (\bibinfo {year} {2014})}\BibitemShut {NoStop}%
\bibitem [{\citenamefont {Vodola}(2015)}]{VodolaThesis}%
  \BibitemOpen
  \bibfield  {author} {\bibinfo {author} {\bibfnamefont {D.}~\bibnamefont
  {Vodola}},\ }\emph {\bibinfo {title} {Correlations and Quantum Dynamics of 1D
  Fermionic Models: New Results for the Kitaev Chain with Long-Range
  Pairing}},\ \href
  {http://amsdottorato.unibo.it/6745/1/vodola_davide_tesi.pdf} {Ph.D. thesis},\
  \bibinfo  {school} {Università di Bologna} (\bibinfo {year} {2015}),\
  \bibinfo {note} {an optional note}\BibitemShut {NoStop}%
\end{thebibliography}%

\appendix




\end{document}